\documentclass[11pt]{article}
\usepackage{preamble}

\title{Quantum Query Complexity of Finding a Tarski Fixed Point on a High-Dimensional Grid}
\author[1,2]{Tongyang Li\thanks{tongyangli@pku.edu.cn}}
\author[3]{Weiran Ma\thanks{2400013063@stu.pku.edu.cn}}
\author[1,2]{Ziyi Yang\thanks{yangziyi@stu.pku.edu.cn}}
\author[3]{Xingyu Zhao\thanks{2300012903@stu.pku.edu.cn}}
\affil[1]{Center on Frontiers of Computing Studies, Peking University, Beijing 100871, China}
\affil[2]{School of Computer Science, Peking University, Beijing 100871, China}
\affil[3]{School of Electronics Engineering and Computer Science, Peking University, Beijing 100871, China}
\date{}

\begin{document}

\maketitle

\thispagestyle{empty}

\begin{abstract}
    The Knaster-Tarski fixed-point theorem states that every monotone function over a complete lattice has a fixed point. Beyond its fundamental role in order theory, the theorem and its algorithmic variants have found broad applications in areas such as economics, game theory, and programming languages. While the query complexity of finding a Tarski fixed point has been extensively studied in classical models, comparatively little is known in the quantum setting.

    We prove an $\Omega(k\log n)$ quantum query lower bound for finding a fixed point of a monotone function on $[n]^k$, using the nonnegative spectral adversary method. In the two extremal regimes $n = 2$ and $k = 1$, our quantum lower bound matches the previous classical lower bounds $\Omega(k)$ and $\Omega(\log n)$, respectively.  For $n, k\geq 2$, our bound improves the best previous classical lower bound when $n < k$ and is within a factor of $\log n / \log k$ compared to the known classical lower bound when $n \geq k$. To construct the adversary matrix, we develop the Tree--Filtration Adversary Method. Besides yielding our lower bound, the method offers a more transparent combinatorial interpretation of the nonnegative spectral adversary method. When the hard instances of a problem admit a tree-like organization and suggest an intuition analogous to classical decision-tree lower bounds, our method provide a promising approach to establishing quantum complexity lower bounds.
\end{abstract}

\section{Introduction}
The Knaster--Tarski fixed-point theorem states that every monotone function $f \colon L \to L$ over a complete lattice $(L,\preceq)$ has a fixed point~\cite{knaster1928theoreme,Tarski1955ALF}. The theorem has extensive
applications in semantics, game theory, and complexity theory. In programming languages, the least and greatest fixed points provide the semantic foundation for recursive programs, monotone data-flow analysis, abstract interpretation, fixed-point logics, verification, and model checking~\cite{kildall1973unified,cousot1977abstract,kozen1983results}. In economics and game theory, Tarski fixed points characterize equilibria in supermodular games and clearing states in financial networks,
making their computation relevant to equilibrium analysis and systemic risk~\cite{topkis1979equilibrium,topkis1998supermodularity,Eisenberg01FinancialNetwork,milgrom1990rationalizability,etessami2019tarski}. In complexity theory, the existence of a $\mathrm{poly}(k, \log n)$ algorithm for the Tarski fixed-point problem on the grid lattice implies that various algorithmic game theory problems in $\textsf{NP}\cap \textsf{coNP}$, such as the Simple Stochastic Game \cite{CONDON1992203} and Arrival \cite{Dohrau2017}, will be in $\textsf{P}$ \cite{etessami2019tarski, gartner_et_al:LIPIcs.ICALP.2021.69}.

The computational complexity of finding Tarski fixed points on finite grids has been extensively studied. For positive integers $n$ and $k$, let $[n]:=\{0,1,\ldots ,n-1\}$ and equip $[n]^k$ with the coordinatewise partial order, where $\boldsymbol{x}\preceq \boldsymbol{y}$ if and only if $x_i\leq y_i$ for every $i\in[k]$. The resulting poset $([n]^k,\preceq)$ forms a complete lattice. In the value-oracle model, an algorithm is given query access to an unknown monotone function $\boldsymbol{f}\colon[n]^k\to[n]^k$, with each query at $\boldsymbol{x}$ returning $\boldsymbol{f}(\boldsymbol{x})$. We denote the problem of finding a fixed point $\boldsymbol{x}^*\in[n]^k$ satisfying $\boldsymbol{f}(\boldsymbol{x}^*)=\boldsymbol{x}^*$ on this lattice using as few queries as possible by \textsc{Tarski}$(n,k)$. Classical research on this problem has proceeded along two main directions: developing efficient algorithms for finding fixed points and constructing hard instances that establish query lower bounds.

In terms of classical algorithms, one of the simplest ways is to iterate $f$ from the bottom (or top) element of $[n]^k$ and find a fixed point using $O(kn)$ queries, since every nonstationary iteration makes
strict progress in at least one coordinate~\cite{etessami2019tarski}. Dang et al.~introduced a recursive coordinate-wise binary-search method using $O((\log n)^k)$ queries for constant $k$~\cite{dang2011computational}. Fearnley et al.~gave an algorithm using $O((\log n)^2)$ queries for $k=3$; together with their decomposition theorem and the bounds in dimensions one and two, this yields an algorithm using $O((\log n)^{\lceil 2k/3\rceil})$ queries for every constant $k$~\cite{dang2011computational,fearnley2022faster}.  Chen and Li subsequently improved the general
bound to $O((\log n)^{\lceil (k+1)/2\rceil})$
~\cite{chen2022improved}.  Haslebacher and Lill later obtained an alternative algorithm with $O((\log n)^2)$ queries for the three-dimensional case
~\cite{haslebacher2026levelset}. Most recently, Chen et al.~gave an algorithm with $O((\log n)^2)$ queries for $k=4$ and, more generally, an
algorithm with $O((\log n)^{\lceil (k-1)/3\rceil+1})$ queries for every constant $k$~\cite{chen2026mystery}.

In terms of classical query lower bounds, Etessami et al.~proved an $\Omega((\log n)^2)$ randomized lower bound for the two-dimensional grid by constructing herringbone functions that encode nested ordered search~\cite{etessami2019tarski}. Br{\^a}nzei et al.~later established the first dimension-sensitive bounds: they characterized the
randomized and deterministic query complexity on the Boolean hypercube as $\Theta(k)$ and proved the general randomized lower bound
$\Omega(k+k\log n/\log k)$~\cite{BPRRandomizedLowerbound}. Their subsequent
multidimensional herringbone construction proved the additional randomized lower bound $\Omega(k(\log n)^2/\log k)$ for $n,k\geq2$~\cite{branzei2025tarski}. In addition, Chen et al.
showed that, in the deterministic black box model, promising the fixed point to be unique does not reduce the query complexity~\cite{chen2026reducing}. Consequently, the
classical complexity is now tight at $\Theta((\log n)^2)$ for $k\in\{2,3,4\}$ and at $\Theta(k)$ when $n=2$, while a substantial gap remains
for general $n$ and $k$.

Despite the extensive classical literature, quantum query complexity of finding a Tarski fixed point has received little attention. To the best of our knowledge, the only prior work is Phillips' recent result for the two-dimensional grid~\cite{phillips2026quantum}. Phillips embedded nested ordered search into a family of herringbone functions and proved an $\Omega((\log n)^2)$ bounded-error quantum lower bound using a composition theorem for the nonnegative spectral adversary method. By monotonicity in the dimension, the same lower bound holds for every $k\geq2$~\cite{phillips2026quantum}. However, the quantum query complexity of this problem on high-dimensional grids remains largely open. Our work aims at deriving quantum query bounds for general $n$ and $k$ that reveal their joint influence on the complexity of finding a Tarski fixed point.

\subsection{Our Contributions}
Our main result is the following theorem.

\begin{theorem}
\label{thm:main-theorem}
    For every constant $0 < \epsilon < 1/2$ and integers $n\geq 2, k\geq 1$,
    $$
    Q_\epsilon(\textsc{Tarski}(n, k)) = \Omega(k\log n).
    $$
\end{theorem}

This new quantum lower bound matches the classical $\Theta(k)$ tight bound when $n = 2$~\cite{BPRRandomizedLowerbound} and the $\Theta(\log n)$ tight bound when $k = 1$~\cite{chang2008complexity,HPJ01BinarySearch}.
For general $n$ and $k$, our quantum lower bound naturally yields the same classical lower bound. For $n,k\geq2$, our bound improves the classical $\Omega(k\log^2 n / \log k)$ lower bound~\cite{branzei2025tarski} when $n< k$ (and has the same order when $n=k$), while for $n\geq k$ it is smaller by a factor of $\Theta(\log n/\log k)$. We summarize the best-known query lower bounds in Table~\ref{tab:lower-bound-comparison}.

\begin{table}[H]
    \centering
    \footnotesize
    \setlength{\tabcolsep}{6pt}
    \renewcommand{\arraystretch}{1.18}
    \begin{tabular}{c|ccc}
        \hline
        Model & Parameter Setting & Best-Known Bounds & Reference \\
        \hline
        \multirow{5}{*}{Classical} & $k = 1$  & $\Theta(\log n)$ & \cite{chang2008complexity} \\
        ~ & $k = 2, 3, 4$ & $\Theta(\log^2 n)$ & \cite{etessami2019tarski}; \cite{fearnley2022faster, chen2026mystery} \\
        ~ & $k > 4, n = 2$ & $\Theta(k)$ & \cite{BPRRandomizedLowerbound}; trivial iterative algorithm \\
        ~ & $k > 4, k \leq n, n > 2$ & $\Omega\left({k\log^2 n}/{\log k}\right)$ & \cite{branzei2025tarski} \\
        ~ & $k > 4, k > n > 2$ & $\Omega(k\log n)$ & \textbf{This work} (implied by Theorem \ref{thm:main-theorem}) \\
        \hline
        \multirow{3}{*}{Quantum} & $k = 1$ & $\Theta(\log n)$ & Reduction to \cite{HPJ01BinarySearch}; \cite{chang2008complexity}\\
        ~ & $2\leq k \leq \log n$ & $\Omega(\log^2 n)$ & \cite{phillips2026quantum} \\
        ~ & $k > \log n$ & $\Omega(k\log n)$ & \textbf{This work} (Theorem \ref{thm:main-theorem}) \\
        \hline
    \end{tabular}

    \caption{Best-known query lower bounds for \textsc{Tarski}$(n,k)$.
    Tight bounds are stated using $\Theta(\cdot)$. If a lower bound is tight, the work establishing the matching upper bound is also included in the table, separated from the lower-bound reference by a semicolon.}
    \label{tab:lower-bound-comparison}
\end{table}

At a high level, we construct a family of hard functions and analyze it using the nonnegative spectral adversary method introduced in \cite{HSAdversary}, in the formulation of \cite{SSAdversary}. Based on the structure of the problem and our hard instances, we organize the hard instance family into a \textit{filtration tree} and define a top-down random walk on this tree that exposes its key structural properties. We then construct an adversary matrix solely from the probability distribution induced by the random walk and use it to establish the desired lower bound.

Motivated by the potential applicability of this approach, we formalize it as the \textit{Tree-Filtration Adversary Method} and develop a general framework in Section~\ref{subsec:general_framework}. The framework exploits the underlying
structure of the problem within the spectral adversary method while avoiding
problem-specific and often cumbersome algebraic calculations. To illustrate
its applicability, we provide in Appendix \ref{apdx:binary-search} a new proof of the \(\Omega(\log n)\) lower bound for one-dimensional binary search that bypasses the spectral analysis of the Hilbert matrix~\cite{HPJ01BinarySearch}.

\subsection{Technical Overview}
We give intuitions and a high-level overview of our quantum query lower bound. In general, we work within the framework of the spectral adversary method, which is reviewed in Section \ref{subsec:prelim_adversary}. Our strategy of constructing and analyzing the adversary matrix may be of independent interest. We derive the matrix canonically for the hard instances organized in a tree structure. By contrast, many earlier applications of the adversary method relied on problem-specific choices of relations, weights, and test vectors, followed in some cases by spectral calculations that obscure the original combinatorial intuition. For example, in ordered search, H{\o}yer, Neerbek, and Shi  \cite{HPJ01BinarySearch} set $\boldsymbol{\Gamma}[\boldsymbol{x},\boldsymbol{y}]=1/|\boldsymbol{x}-\boldsymbol{y}|$ and bound the query matrices by identifying their off-diagonal blocks with finite sections of the Hilbert matrix, whose spectral norm is at most $\pi$; this analytic step has little direct connection to the structure of ordered search. The argument is powerful, but the relevant algebraic objects are designed and analyzed separately. Our method is analogous in spirit to \cite{AmbainisAndOr}, which shows that if the hard instance can be partitioned into two disjoint sets $X$ and $Y$ such that $f(\boldsymbol{x})\ne f(\boldsymbol{y})$ for every $\boldsymbol{x}\in X$ and $\boldsymbol{y}\in Y$, then a quantum query lower bound can be derived from certain combinatorial parameters of a relation $R\subseteq X\times Y$. In comparison, our method shows that when the hard instances admit a tree-like organization, a quantum query lower bound can be derived from two probabilistic properties. Thus, our method has the potential to provide unified proofs of quantum lower bounds for problems that admit intuitions analogous to classical decision-tree lower bounds.

\paragraph{The tree-filtration adversary method.}
Let $\mathcal F$ be a family of hard instances for a function problem. We assume that the instances in $\mathcal F$ have been organized into a rooted tree $\mathcal{T} = (\mathcal{V}, \mathcal{E}, r_{\mathcal{T}})$ whose leaves all have the same depth $L$, with each leaf corresponding to exactly one hard instance. An internal node $u$ therefore represents the set of all instances associated with the leaves in the subtree rooted at $u$. Moving down the tree reveals progressively more information about the hidden instance, until the instance is completely determined at a leaf.

The construction is guided by the same intuition as a classical decision-tree lower bound. A useful decision tree should have two properties:
\begin{enumerate}
    \item Every proper subtree should still contain instances with different answers; and
    \item The instances within a subtree should remain difficult to distinguish using only a small number of queries.
\end{enumerate}

To give a unified formulation of the second property, we associate with every node $u$ a set $\hat D(u)$ of query locations, determined solely by $u$, called the \emph{simplified difference criterion} (Definition \ref{def:difference-criterion}). This set provides a sufficient certificate for the possible distinguishability of instances within the subtree: for any two instances $\boldsymbol{x}$ and $\boldsymbol{y}$ below $u$,
$$
    \boldsymbol{x}_i\neq \boldsymbol{y}_i \quad\Longrightarrow\quad i\in\hat D(u).
$$
We further require these sets to be \emph{nested}: if $v$ is a descendant of $u$ on $\mathcal{T}$, then $\hat D(v)\subseteq\hat D(u)$. Hence, as more information about the hidden instance is exposed, the set of queries that may still distinguish the remaining candidates becomes progressively smaller. This nested family of potential difference sets is the \textit{filtration} in the name of our method.

We probe the tree using a top-down random walk $\mathcal A$. At every internal node, $\mathcal A$ chooses one of its children according to a prescribed probability distribution. The walk induces a distribution over the leaves of the tree and, after conditioning on any node $u$, a distribution over the hard instances in the subtree rooted at $u$. We require this walk to satisfy two complementary properties.

\begin{enumerate}
    \item \textbf{Anti-concentration (Definition \ref{def:anti-concentration}).} For a node $u$ at depth $t$, a random continuation of $\mathcal A$ should retain a noticeable probability of producing an answer different from that of any prescribed leaf below $u$. Formally, this probability is required to be at least $a_t$. Intuitively, even after the information represented by $u$ has been revealed, the output remains dispersed among the instances in its subtree.

    \item \textbf{Geometric-indistinguishability (Definition \ref{def:geometric-indistinguishable}).} Fix a query location $i$. Suppose the walk starts from a node $u$ at depth $s$ and reaches a random descendant $V$ at depth $t>s$. We require
    $$
        \Pr\left[i\in\hat D(V)\middle| \mathcal A\text{ starts at }u\right]
        \leq r^{t-s},
    $$
    for some constant $r<1$. Since $\hat D(V)$ contains every location at which instances below $V$ may disagree, this condition says that the distinguishing power of any fixed query decays geometrically as the walk descends the tree.
\end{enumerate}

We next construct an adversary matrix directly from the tree and the distribution induced by $\mathcal A$. Let $\mathrm{ins}(\ell)$ denote the hard instance associated with a leaf $\ell$, and $\mathrm{ins}^{-1}$ its inverse map, i.e. the corresponding leaf to a hard instance. Write $\Pr_{\mathcal A}[u]$ for the probability that the walk reaches a node $u$, and $\phi_u[\boldsymbol{x}] = \sqrt{\frac{\Pr_{\mathcal{A}}[\mathrm{ins}^{-1}(\boldsymbol{x})]}{\Pr_{\mathcal{A}}[u]}}\mathds{I}\{\mathrm{ins}^{-1}(\boldsymbol{x})\in \mathcal{L}(\mathrm{Sub}(u))\}$. For every node $u$, define
$$
\boldsymbol{\Gamma}_u[\boldsymbol{x},\boldsymbol{y}]=\phi_{u}[\boldsymbol{x}]\phi_u[\boldsymbol{y}]\mathds{I}\{f(\boldsymbol{x}) \ne f(\boldsymbol{y})\}
$$
and let $\boldsymbol{\Gamma}=\sum_{u\in \mathcal{V}\setminus \mathcal{L}(\mathcal{T})}\boldsymbol{\Gamma}_u$. 

The two properties of $\mathcal A$ control the numerator and denominator of the spectral adversary bound as we prove in Lemma \ref{lem:numerator-lb} and Lemma \ref{lem:denominator-ub}. The proof of the lower bound is therefore reduced to finding a suitable tree, a nested difference criterion $\hat D$, and a random walk satisfying the two probabilistic conditions above.

\paragraph{The hard-instance family.}
We now describe a family of Tarski instances that admits exactly such a tree. Each hard instance is built around a \emph{spine} consisting of $\Theta(k\log n)$ vertices:
$$
\mathbf 0=\boldsymbol{x}_0\preceq \boldsymbol{x}_1\preceq\cdots\preceq \boldsymbol{x}_{|\textsf{Tok}(n,k)|}=B\mathbf 1.
$$

The sequence $\{\boldsymbol{x}_i-\boldsymbol{x}_{i-1}\}_{i=1}^{|\textsf{Tok}(n,k)|}$ forms a \textit{balanced} permutation (Definition \ref{def:balanced-permutations}) of a set of vectors that we call \emph{tokens} (Definition \ref{def:tokens}). Each token has the form $(j, r)$ and represents the vector $4^r \boldsymbol{e}_j$, where $r\leq U$, $U$ is maximal subject to $1 + \sum_{r = 0}^{U}4^r \leq n$ and $\boldsymbol{e}_j$ denotes the $j$th standard basis vector. Thus, a token specifies both a coordinate direction and a geometric scale. The $k$ possible directions and the $\Theta(\log n)$ possible scales together yield $\Theta(k\log n)$ tokens.

We designate the midpoint of the spine as the unique fixed point. Every other point on the spine is mapped to the adjacent spine vertex that is one step closer to the midpoint. To define the function on the rest of the grid, we first project each lattice point $\boldsymbol{v}$ onto the spine and then move one step toward the midpoint. For a point $\boldsymbol{v}$ in the lower half of the lattice, the projection is the greatest spine vertex $\boldsymbol{x}_i$ satisfying $\boldsymbol{x}_i\preceq \boldsymbol{v}$; for a point in the upper half, it is the least spine vertex $\boldsymbol{x}_i$ satisfying $\boldsymbol{v}\preceq \boldsymbol{x}_i$. Here the two halves are determined by the $\ell_1$-norm of $\boldsymbol{v}$.

We show that every function in the resulting family is monotone (Lemma \ref{lem:monotonicity}) and has a unique fixed point (Lemma \ref{lem:unique_fixed_point}), namely the midpoint of its spine.

\paragraph{The prefix-suffix tree.}
We organize these instances according to the prefix--suffix hierarchy of their token permutations. A node at depth $t$ specifies the first $t$ tokens and the last $t$ tokens, with the latter recorded in reverse order. Its subtree contains all balanced token permutations consistent with this partial assignment. At the final depth, the prefix and suffix determine the complete permutation, so the leaves are in bijection with the hard instances and all occur at the same depth.

This tree realizes the two classical decision-tree intuitions described above. First, every proper subtree hides many possible fixed points, and in particular at least two. Second, queries on most lattice points reveal very little about the unrevealed part of the spine (see Example \ref{eg:function-2-8} for some numerical results). The corresponding simplified difference criterion is
$$
\hat{D}(u_{\boldsymbol{p}, \boldsymbol{q}}) = \{\boldsymbol{x} : \norm{\boldsymbol{x}} \leq W / 2 \wedge \boldsymbol{c}(\boldsymbol{p}) \preceq \boldsymbol{x}\} \cup \{\boldsymbol{x} : \norm{\boldsymbol{x}} > W / 2 \wedge \boldsymbol{x} \preceq B\boldsymbol{1} - \boldsymbol{c}(\boldsymbol{q})\}.
$$

It remains to choose the appropriate distribution on the tree. Our random walk $\mathcal A$, described formally in Definition \ref{def:distinguishing-walk}, proceeds as follows. At a node $(\boldsymbol{p},\boldsymbol{q})$, it selects a scale $r$ with probability proportional to the number of unused tokens at that scale multiplied by their length $4^r$. It then chooses two distinct unused tokens of scale $r$ uniformly at random without replacement, appending one to the prefix and the other to the suffix.

This walk satisfies both requirements of the general framework. Its induced distribution is $1/2$-anti-concentrated and $2/3$-geometrically indistinguishable. Finally, the depth of the prefix-suffix tree is $\frac{|\textsf{Tok}(n,k)|}{2}=\Theta(k\log n)$. Applying the tree-filtration adversary method with $a_t=1/2$ at every non-leaf depth and $r=2/3$ therefore gives
$$
    Q_\epsilon(\textsc{Tarski}(n,k))=\Omega(k\log n),
$$
which proves our main theorem.

\subsection{Paper Organization}

In Section \ref{sec:preliminaries}, we provide some definitions and notation for sets, sequences, permutations and quantum algorithms, and review necessary preliminaries in linear algebra and the spectral adversary method. In Section \ref{sec:hard_instance}, we present our hard instance construction and prove its properties. In Section \ref{sec:adversary_method}, we present the high-level idea of constructing our adversary matrix, and prove the main theorem. Finally, we conclude with a discussion and several open problems in Section \ref{sec:conclusion}.

\section{Preliminaries} \label{sec:preliminaries}

\subsection{Notations}

We begin by introducing the notation used throughout the paper. Sequences and vectors are denoted by bold lowercase letters and linear maps by bold uppercase letters.

Let $\mathds{I}\{P\}$ denote the indicator of the proposition $P$; that is, it equals $1$ if $P$ is true and $0$ otherwise. For every positive integer $n$, let $[n] := \{0, 1, \ldots , n - 1\}$. For a set $S$ and a nonnegative integer $\ell$, let $S^\ell$ denote the set of all sequences of length $\ell$ over $S$ and define $S^* := \cup_{\ell\geq 0}S^{\ell}$. For a sequence $\boldsymbol{p}$, we denote its length by $|\boldsymbol{p}|$. For $\boldsymbol{p},\boldsymbol{q}\in S^{*}$, their concatenation, denoted by $\boldsymbol{p}\mathbin{\|}\boldsymbol{q}$, is the sequence of length $|\boldsymbol{p}|+|\boldsymbol{q}|$ defined by \[ (\boldsymbol{p}\mathbin{\|}\boldsymbol{q})_i := \begin{cases} p_i, & 0\leq i<|\boldsymbol{p}|,\\ q_{i-|\boldsymbol{p}|}, & |\boldsymbol{p}|\leq i<|\boldsymbol{p}|+|\boldsymbol{q}|. \end{cases} \]

Let $S$ be a finite set. For every integer $0\leq \ell \le|S|$, we define the set of length $\ell$ partial permutations of $S$ as $\textsf{P}_S^\ell := \{ \boldsymbol{p}\in S^\ell : \forall 0 \leq i < j < \ell, p_i\ne p_j\}.$ For brevity, write $\textsf{P}_S := \textsf{P}_S^{|S|}$ for the set of all permutations of $S$. It is immediate that $|\textsf{P}_S^\ell| = |S|!/(|S| - \ell)!$. In our proof, we will sometimes consider permutations whose prefixes and suffixes are fixed. Such a restriction is valid only when the prescribed prefix and suffix contain no common element. More precisely, for nonnegative integers $\ell_1$ and $\ell_2$ satisfying $\ell_1+\ell_2\leq|{S}|$, define
$$
\textsf{PPair}_{S}^{\ell_1, \ell_2} = \{(\boldsymbol{p}, \boldsymbol{q}) \in \textsf{P}_S^{\ell_1}\times \textsf{P}_S^{\ell_2} : \boldsymbol{p} \mathbin{\|} \boldsymbol{q}\in \textsf{P}_S^{\ell_1 + \ell_2}\}.
$$ 
We refer to the elements of $\textsf{PPair}_S^{\ell_1,\ell_2}$ as \emph{valid permutation pairs}. Given a valid permutation pair $(\boldsymbol{p}, \boldsymbol{q})\in \textsf{PPair}_S^{\ell_1, \ell_2}$, we can define the \emph{prefix-suffix restricted permutations} as follows:
$$
\textsf{P}_S|_{\boldsymbol{p}, \boldsymbol{q}} = \left\{\boldsymbol{\pi} \in \textsf{P}_S : \begin{aligned}
    \forall i = 0, 1, \ldots , \ell_1 - 1, & \pi_i = p_i \\
    \forall i = 0, 1, \ldots , \ell_2 - 1, & \pi_{|S| - i - 1} = q_i
\end{aligned}\right\}.
$$

Let $\mathcal{T} = (\mathcal{V}, \mathcal{E}, r)$ be a tree rooted at $r\in \mathcal{V}$. Let $\mathcal{L}(\mathcal{T})$ denote the set of all leaf nodes in $\mathcal{T}$. For each node $u$, denote the depth of $u$ by $\mathrm{dep}(u)$, the subtree of $u$ in $\mathcal{T}$ by $\mathrm{Sub}(u)$, all children of $u$ by $\mathrm{ch}(u)$, and the ancestor of $u$ at depth $d$ by $\mathrm{anc}(u, d)$.
\subsection{Basic Properties of the Lattice} \label{subsec:lattice}

Let $(L,\preceq)$ be a partially ordered set.  It is called a complete lattice if every subset $S\subseteq L$ has both a supremum and an infimum. In this work, we consider the $k$-dimensional grid $[n]^k$. A coordinatewise partial order $\preceq$ on $[n]^k$ is defined by $\forall \boldsymbol{a}, \boldsymbol{b}\in [n]^k, \boldsymbol{a}\preceq \boldsymbol{b} \Leftrightarrow \forall i\in [k], a_i\leq b_i$. Clearly, $([n]^k, \preceq)$ forms a complete lattice. For any function $\boldsymbol{f} \colon [n]^k\rightarrow [n]^k$, say $\boldsymbol{f}$ is monotone if and only if $\forall \boldsymbol{a}, \boldsymbol{b}\in [n]^k, \boldsymbol{a}\preceq \boldsymbol{b}\rightarrow \boldsymbol{f}(\boldsymbol{a}) \preceq \boldsymbol{f}(\boldsymbol{b}).$ The composition of two monotone functions is also monotone.

\begin{lemma} \label{lem:composition_monotonicity}
    Let $\boldsymbol{f}, \boldsymbol{g}$ be two monotone functions on $[n]^k$. Then $\boldsymbol{f}\circ \boldsymbol{g}$ is monotone.
\end{lemma}

\begin{proof}
    It follows directly from the definition of monotone functions that
    $$
    \boldsymbol{x}\preceq \boldsymbol{y}\Rightarrow \boldsymbol{g}(\boldsymbol{x})\preceq \boldsymbol{g}(\boldsymbol{y})\Rightarrow \boldsymbol{f}(\boldsymbol{g}(\boldsymbol{x})) \preceq \boldsymbol{f}(\boldsymbol{g}(\boldsymbol{y})),
    $$ which shows $\boldsymbol{f}(\boldsymbol{g}(\cdot))$ is monotone.
\end{proof}

\subsection{Linear Algebra}

A central component of our analysis is bounding the spectral norms of certain matrices. For a matrix $\mathbf{A}\in\R^{n\times n}$, its spectral norm is defined by \[ \norm{\mathbf{A}} := \max_{\boldsymbol{x}\in\R^n:\,\norm{\boldsymbol{x}}_2=1} \norm{\mathbf{A}\boldsymbol{x}}_2. \]

We will use the following standard facts. If $\mathbf{A}$ is real symmetric, then $\norm{\mathbf{A}} = \sup_{\boldsymbol{x} : \norm{\boldsymbol{x}}_2 = 1} |\boldsymbol{x}^\top \mathbf{A}\boldsymbol{x}|$ and $\norm{\mathbf{A}} \leq \norm{\mathbf{A}}_1$, where $\norm{\cdot}_1$ is the induced $1$-norm.

The following two lemmas are useful for bounding the spectral norm of matrices.

\begin{lemma}[Entrywise monotonicity of the spectral norm] \label{lem:spectral-norm-monotonicity}
    Let $\mathbf{A}, \mathbf{B}\in \R^{n\times n}$ be two real symmetric matrices, satisfying $\mathbf{B}\geq \mathbf{A}\geq \mathbf{0}$ (where $\geq$ means entry-wise greater or equal). Then
    $$
    \norm{\mathbf{B}} \geq \norm{\mathbf{A}}.
    $$
\end{lemma}

\begin{proof}
    For every unit vector $\boldsymbol{v}\in \R^n$, write $|\boldsymbol{v}|$ for the entry-wise absolute value of $\boldsymbol{v}$. Then since $\mathbf{A}, \mathbf{B}$ are real symmetric,
    $$
    \norm{\mathbf{A}} = \sup_{\boldsymbol{v} :  \norm{\boldsymbol{v}}_2 = 1}|\boldsymbol{v}^\top\mathbf{A}\boldsymbol{v}| \leq \sup_{\boldsymbol{v}:\norm{\boldsymbol{v}}_2 = 1}|\boldsymbol{v}|^\top\mathbf{A}|\boldsymbol{v}| \leq \sup_{\boldsymbol{v}:\norm{\boldsymbol{v}}_2 = 1}|\boldsymbol{v}|^\top\mathbf{B}|\boldsymbol{v}| \leq \norm{\mathbf{B}}.
    $$
\end{proof}

\begin{lemma}
    \label{lem:summation-spectral-upperbound} Let $\mathbf{A}_0,\ldots ,\mathbf{A}_{m - 1}\in\R^{n\times n}$ be real symmetric idempotent matrices, and let $\mathbf{B}:=\sum_{i=0}^{m - 1}\mathbf{A}_i$. Suppose that there exists $r\in(0,1)$ such that \[ \norm{\mathbf{A}_i\mathbf{A}_j} \leq r^{|i-j|} \qquad \text{for all } i,j\in[m]. \] Then $\norm{\mathbf{B}}\leq(1+r)/(1-r)$.
\end{lemma}

\begin{proof}
    Define the linear operator $\mathcal{A}\colon \R^n\rightarrow \R^{mn}: \mathcal{A}\boldsymbol{v} :=(\mathbf{A}_0\boldsymbol{v},\ldots , \mathbf{A}_{m - 1}\boldsymbol{v})^{\top}$. Then $\mathcal{A}^\top\mathcal{A} = \mathbf{B}$ since for any $\boldsymbol{u}, \boldsymbol{v}$, 
    $$\boldsymbol{u}^\top\mathcal{A}^\top\mathcal{A}\boldsymbol{v} = \sum_{i = 0}^{m - 1} \boldsymbol{u}^\top \mathbf{A}_i^2 \boldsymbol{v} = \boldsymbol{u}^\top\left(\sum_{i = 0}^{m - 1} \mathbf{A}_i\right)\boldsymbol{v} = \boldsymbol{u}^\top\mathbf{B}\boldsymbol{v}.
    $$

    Therefore, $\norm{\mathbf{B}} = \norm{\mathcal{A}^\top\mathcal{A}} = \norm{\mathcal{A}\mathcal{A}^\top}$. The operator $\mathcal{A}\mathcal{A}^{\top}$ is an $m\times m$ block matrix whose $(i,j)$-block is $\mathbf{A}_i\mathbf{A}_j$. Define $\mathbf{K}\in\R^{m\times m}$ by $K_{ij}:=\norm{\mathbf{A}_i\mathbf{A}_j}$. It is easy to see that $\mathbf{K}$ is symmetric. The standard block-matrix norm comparison gives $$\norm{\mathcal{A}\mathcal{A}^\top} \leq \norm{\mathbf{K}}\leq \norm{\mathbf{K}}_1 = \max_{j\in[m]}\sum_{i=0}^{m - 1} K_{ij} \leq 1+2\sum_{d=1}^{\infty}r^d= \frac{1 + r}{1 - r}.
    $$
\end{proof}

\subsection{Computational Model} \label{subsec:computation_model}

In this paper, we work in the quantum query model. Let $\Sigma$ be a finite alphabet equipped
with an abelian group operation $\oplus$, let $\Pi$ be a finite output
alphabet, let $I$ be an index set of cardinality $n$, and let $f\colon S\to\Pi$ be a partial function, where $S\subseteq \Sigma^I$. The input
$\boldsymbol{x}\in S$ is not given explicitly but is accessed through
the oracle $\mathbf{O}_{\boldsymbol{x}}$, defined by $\mathbf{O}_{\boldsymbol{x}}\ket{i,z} = \ket{i,z\oplus x_i}$ for every $i\in I$ and $z\in\Sigma$. 

A $T$-query quantum algorithm starts in a fixed state $\ket{\psi_0}$ that is
independent of $\boldsymbol{x}$ and applies input-independent unitary
operators $\mathbf{U}_0,\ldots ,\mathbf{U}_T$, interleaved with $T$ applications
of $\mathbf{O}_{\boldsymbol{x}}$. Its final state is therefore
\[
\ket{\psi_{\boldsymbol{x}}^T}
=
\mathbf{U}_T\mathbf{O}_{\boldsymbol{x}}
\mathbf{U}_{T-1}\cdots
\mathbf{O}_{\boldsymbol{x}}\mathbf{U}_0\ket{\psi_0},
\]
where the identity action of the oracle on the ancillary registers is
implicit. The algorithm then measures a designated output register whose
computational-basis states are indexed by $\Pi$.

The algorithm computes $f$ with error at most $\epsilon$ if, for every
$\boldsymbol{x}\in S$, the measurement outputs $f(\boldsymbol{x})$ with
probability at least $1-\epsilon$. The bounded-error quantum query
complexity of $f$, denoted by $Q_\epsilon(f)$, is the minimum nonnegative
integer $T$ for which such a $T$-query quantum algorithm exists.

\subsection{Spectral Adversary Method}
\label{subsec:prelim_adversary}

The quantum adversary methods are a family of powerful tools for deriving quantum query lower bounds for computing functions in the quantum query model.

The nonnegative spectral adversary method was introduced in \cite{HSAdversary}. It was proved to be equivalent to other nonnegative adversary methods in \cite{SSAdversary}. Lower bounds for many problems were established using the adversary methods, for instance, an $\Omega(\sqrt{n})$ lower bound for $\mathrm{AND}, \mathrm{OR}$ in \cite{AmbainisAndOr}, and an $\Omega(\log n)$ lower bound for ordered search in \cite{HPJ01BinarySearch}. The nonnegative spectral adversary method was proved to be weaker than the adversary method with negative weights in \cite{HLS07Negative}.

Throughout this section, let $\Pi, \Sigma$ denote two finite alphabets, let $I$ be an index set of cardinality $n$, and let $f\colon S\rightarrow \Pi$ be a partial function from $\Sigma^I$ to $\Pi$, where $S\subseteq \Sigma^I$. We define two types of matrices used in the argument as follows.

\begin{definition}[Nonnegative Adversary Matrix]
A matrix $\boldsymbol{\Gamma}\in \R_{\geq 0}^{|S|\times |S|}$ is called a \emph{nonnegative adversary matrix} if it is symmetric and satisfies the following property:
$$
\forall x, y, f(x)= f(y)\rightarrow \boldsymbol{\Gamma}[x, y] = 0.
$$
\end{definition}

\begin{definition}[Distinguisher Matrix]
For every $i\in I$, the $i$-th \emph{distinguisher matrix} for $f$ is $\boldsymbol{\Delta}_i \in \{0, 1\}^{|S|\times |S|}$ with entries:
$$
\Delta_i[x, y] = \mathbf{1}\{ x_i\ne y_i \}.
$$
\end{definition}

Any such $\boldsymbol{\Gamma}$ yields a lower bound on the quantum query complexity of computing $f$ in the quantum query model, as stated in the following theorem. Intuitively, a good adversary matrix should assign $\boldsymbol{\Gamma}[x, y]$ a magnitude that reflects the difficulty of distinguishing the two instances $x, y$.

\begin{theorem}[See Theorem 3.1 of \cite{SSAdversary}]
    For any nonnegative adversary matrix $\boldsymbol{\Gamma}$ for $f$, we have
    $$
    Q_\epsilon(f)\geq \left(1 - 2\sqrt{\epsilon(1 - \epsilon)}\right)\frac{\norm{\boldsymbol{\Gamma}}}{\max_{i\in I}\norm{\boldsymbol{\Gamma}\circ\boldsymbol{\Delta}_{i}}},
    $$
    where $\boldsymbol{\Delta}_i$ is the $i$-th distinguisher matrix for $f$ and $\circ$ denotes the Hadamard product.
\end{theorem}

\section{Hard Instance} \label{sec:hard_instance}

The classical lower-bound constructions yield two main families of hard instances. For the first family, introduced in~\cite{BPRRandomizedLowerbound}, we design an efficient quantum algorithm using only $\widetilde{O}(\sqrt{k})$ queries by reducing these instances to the quantum wildcard-search problem~\cite{WildcardSearch}. The details are shown in Appendix~\ref{apdx:bpr25-upperbound}. For the second family~\cite{branzei2025tarski}, no quantum upper bound is currently known; however, its intricate combinatorial structure makes standard techniques for proving quantum query lower bounds difficult to apply. We therefore introduce a new family of hard instances whose structure is tailored to our spectral adversary method.

It was proved in~\cite{phillips2026quantum} that the query complexity of
$\textsc{Tarski}(n,k)$ is monotone in both parameters. In the nontrivial regime
$n,k\geq2$, we may therefore replace $k$ by the largest even integer not
exceeding it and restrict the grid width to $B_U+1$, where
$B_U:=\sum_{r=0}^{U}4^r$ and $U$ is maximal subject to $B_U+1\leq n$. These
replacements preserve $k$ and $\log n$ up to constant factors. Relabeling the
resulting parameters, we henceforth assume that $k$ is even and that $n-1=B:=\sum_{r=0}^{U}4^r=1+4+16+\cdots+4^U$. For a lattice point $\boldsymbol{v}$, the notation
$\norm{\boldsymbol{v}}$ denotes its $\ell_1$-norm
$\sum_{j\in[k]}v_j$.

At a high level, each hard instance is built around a monotone path, called its
\emph{spine} (Definition \ref{def:spine}), from $\boldsymbol{0}$ to $B\boldsymbol{1}$. Set $W := \norm{B\boldsymbol{1}}=kB$. Every lattice point is
projected onto the spine (Definition \ref{def:projection}) and then moved by one spine edge toward the central vertex (Definition \ref{def:hard_function}), which is the unique fixed point. The step lengths
$1,4,\ldots ,4^U$ provide $\Theta(\log n)$ geometric scales; the bound
$\sum_{4^r\leq A}4^r\leq\frac{4}{3}A$ is the reason for choosing powers of
four and will be used in the adversary analysis.

\begin{definition}[Tokens] \label{def:tokens}
Let $\textsf{T}:=\textsf{Tok}(n,k):=[k]\times[U+1]$. For each token
$t=(j_t,r_t)\in\textsf{T}$, call $j_t$ its \emph{direction} and $r_t$ its
\emph{layer}, and associate with it a vector $\boldsymbol{v}_t:=4^{r_t}\boldsymbol{e}_{j_t}$, where $\boldsymbol{e}_{j_t}$ is the $j_t$-th standard basis vector.
\end{definition}

The token set has size $N:=|\textsf{T}|=k(U+1)$. Since $k$ is even, write
$N=2m$. With each partial token permutation
$\boldsymbol{p}$, we associate two summary vectors. The \emph{position vector}
$\boldsymbol{c}(\boldsymbol{p})$ is the lattice point reached by applying the
moves in $\boldsymbol{p}$, whereas the \emph{layer-count vector}
$\boldsymbol{d}(\boldsymbol{p})$ records how many tokens from each layer occur
in $\boldsymbol{p}$. We define these two vectors formally as follows.

\begin{definition}[Position and layer-count vectors]
Let $0\leq \ell\leq N$ and
$\boldsymbol{p}=(p_0,\ldots ,p_{\ell-1})
\in\textsf{P}_{\textsf{T}}^\ell$. Define
\[
\boldsymbol{c}(\boldsymbol{p})
:=\sum_{i=0}^{\ell-1}\boldsymbol{v}_{p_i}\in[n]^k,
\]
and define $\boldsymbol{d}(\boldsymbol{p})\in[k+1]^{U+1}$ by
\[
\boldsymbol{d}(\boldsymbol{p})_r
:=\bigl|\{i\in[\ell]:r_{p_i}=r\}\bigr|,
\qquad r\in[U+1].
\]
\end{definition}

Every full token permutation orders all $N$ moves and hence determines a spine.

\begin{definition}[Spine] \label{def:spine}
Let $\pi=(\pi_0,\ldots ,\pi_{N-1})\in\textsf{P}_{\textsf{T}}$. For
$i\in[N+1]$, let $\pi_{<i}:=(\pi_0,\ldots ,\pi_{i-1})$, with
$\pi_{<0}:=\varepsilon$, and define $\boldsymbol{x}_i^\pi:=\boldsymbol{c}(\pi_{<i}).$ The spine associated with $\pi$ is
$\mathcal{X}_\pi:=\{\boldsymbol{x}_i^\pi:i\in[N+1]\}$. Its vertices are
distinct and satisfy
\[
\boldsymbol{0}=\boldsymbol{x}_0^\pi
\prec\boldsymbol{x}_1^\pi\prec\cdots
\prec\boldsymbol{x}_N^\pi=B\boldsymbol{1}.
\]
\end{definition}

We next project the lattice onto the spine. The auxiliary map $\rho_\pi$ records
the index of the projected vertex.

\begin{definition}[Projection onto a spine] \label{def:projection}
For $\pi\in\textsf{P}_{\textsf{T}}$ and $\boldsymbol{v}\in[n]^k$, define
\[
\rho_\pi(\boldsymbol{v}):=
\begin{cases}
\max\{i\in[N+1]:\boldsymbol{x}_i^\pi\preceq\boldsymbol{v}\},
    & \norm{\boldsymbol{v}}\leq W/2,\\
\min\{i\in[N+1]:\boldsymbol{v}\preceq\boldsymbol{x}_i^\pi\},
    & \norm{\boldsymbol{v}}>W/2.
\end{cases}
\]
The projection $\Pi_\pi\colon[n]^k\to\mathcal{X}_\pi$ is $\Pi_\pi(\boldsymbol{v}):=\boldsymbol{x}_{\rho_\pi(\boldsymbol{v})}^\pi.$
The two index sets above are nonempty because they contain $0$ and $N$,
respectively.
\end{definition}

Define $h\colon[N+1]\to[N+1]$ by
\[
h(i):=
\begin{cases}
i+1, & i<m,\\
m,   & i=m,\\
i-1, & i>m.
\end{cases}
\]
Thus, $h$ moves every spine index one step toward $m$ and fixes $m$.

\begin{definition}[Hard function] \label{def:hard_function}
For $\pi\in\textsf{P}_{\textsf{T}}$, define
$g_\pi\colon\mathcal{X}_\pi\to\mathcal{X}_\pi$ by $g_\pi(\boldsymbol{x}_i^\pi):=\boldsymbol{x}_{h(i)}^\pi$, $i\in[N+1]$, and define $\boldsymbol{f}_\pi\colon[n]^k\to[n]^k$ by
\[
\boldsymbol{f}_\pi:=g_\pi\circ\Pi_\pi.
\]
\end{definition}

For the adversary construction, we restrict attention to permutations whose
prefix and reversed suffix pair tokens from the same layer. For
$\pi\in\textsf{P}_{\textsf{T}}$ and $0\leq\ell\leq m$, write $\boldsymbol{ps}_\ell(\pi)
:=\bigl(\boldsymbol{p}_\ell(\pi),\boldsymbol{q}_\ell(\pi)\bigr),$ where $\boldsymbol{p}_\ell(\pi):=(\pi_0,\ldots ,\pi_{\ell-1})$ records the length-$\ell$ prefix of $\pi$,
$\boldsymbol{q}_\ell(\pi):=(\pi_{N-1},\ldots ,\pi_{N-\ell})$ records the length-$\ell$ suffix of $\pi$ in reversed order. Both sequences are empty when $\ell=0$.

\begin{definition}[Balanced prefix--suffix pairs]
For $0\leq\ell\leq m$, define
\[
\textsf{SPair}_{\textsf{T}}^\ell
:=\left\{(\boldsymbol{p},\boldsymbol{q})
\in\textsf{PPair}_{\textsf{T}}^{\ell,\ell}:
r_{p_i}=r_{q_i}\text{ for every }i\in[\ell]\right\}.
\]
\end{definition}

\begin{definition}[Balanced permutations] \label{def:balanced-permutations}
Define $\textsf{BP}_{\textsf{T}}
:=\left\{\pi\in\textsf{P}_{\textsf{T}}:
\boldsymbol{ps}_m(\pi)\in\textsf{SPair}_{\textsf{T}}^m\right\}.$

\end{definition}

If $\pi$ is balanced, each half contains exactly $k/2$ tokens from every
layer. Consequently, $\norm{\boldsymbol{x}_m^\pi}
=\sum_{r=0}^{U}4^rk/2
=W/2$.

\begin{definition}[Hard instance family] \label{def:hard-instance-family}
    Define $\mathcal{F}_{n, k} = \{\boldsymbol{f}_\pi : \pi \in \textsf{BP}_{\textsf{T}}\}.$
\end{definition}

The map $\pi\mapsto \boldsymbol{f}_\pi$ is injective. Indeed, the orbit of
$\boldsymbol{0}$ under $\boldsymbol{f}_\pi$ recovers
$\boldsymbol{x}_0^\pi,\ldots ,\boldsymbol{x}_m^\pi$, while the orbit of
$B\boldsymbol{1}$ recovers
$\boldsymbol{x}_N^\pi,\ldots ,\boldsymbol{x}_m^\pi$. The difference between each pair of consecutive vertices on the spine then identifies the corresponding token. 

To illustrate the definitions above, we visualize a hard instance on $[22]^2$ in Example \ref{eg:function-21-2}.

\begin{example} \label{eg:function-21-2}
    \normalfont
    Consider $k = 2$ and $n = 1 + \sum_{r = 0}^{U} 4^r = 22$, where $U = 2$. Then $\textsf{Tok}(n, k) = [2]\times [3]$. Take the following permutation of $\textsf{Tok}(n, k)$:
    $$
    \pi = [(1, 2), (0, 0), (1, 1), (0, 1), (1, 0), (0, 2)].
    $$

    Then $\pi$ is balanced since $r_{\pi_0} = r_{\pi_5} = 2, r_{\pi_1} = r_{\pi_4} = 0$ and $r_{\pi_2} = r_{\pi_3} = 1$. The six vectors corresponding to $\pi_0, \ldots, \pi_5$ are
    $$
    16\boldsymbol{e}_1, ~\boldsymbol{e}_0, ~4\boldsymbol{e}_1, ~4\boldsymbol{e}_0, ~\boldsymbol{e}_1, ~16\boldsymbol{e}_0.
    $$
    Taking prefix sums of this sequence gives the spine associated with $\pi$:
    $$
    \mathcal{X}_\pi = \{(0, 0), (0, 16), (1, 16), (1, 20), (5, 20), (5, 21), (21, 21)\}.
    $$
    We represent the $\boldsymbol{x}\rightarrow \boldsymbol{f}(\boldsymbol{x})$ relationship on the spine by black arrows. For other points on the grid, we color the $\boldsymbol{x}\rightarrow \boldsymbol{f}(\boldsymbol{x})$ arrows according to the color scale: the color intensity encodes the distance of $\boldsymbol{f}_\pi(\boldsymbol{x})$ from the fixed point; darker color indicates smaller distance. The unique fixed point is marked by $\color{red} \star$. The resulting function is illustrated in Figure \ref{fig:function-21-2}.

    \begin{figure}[H]
        \centering
        \includegraphics[width=0.5\linewidth]{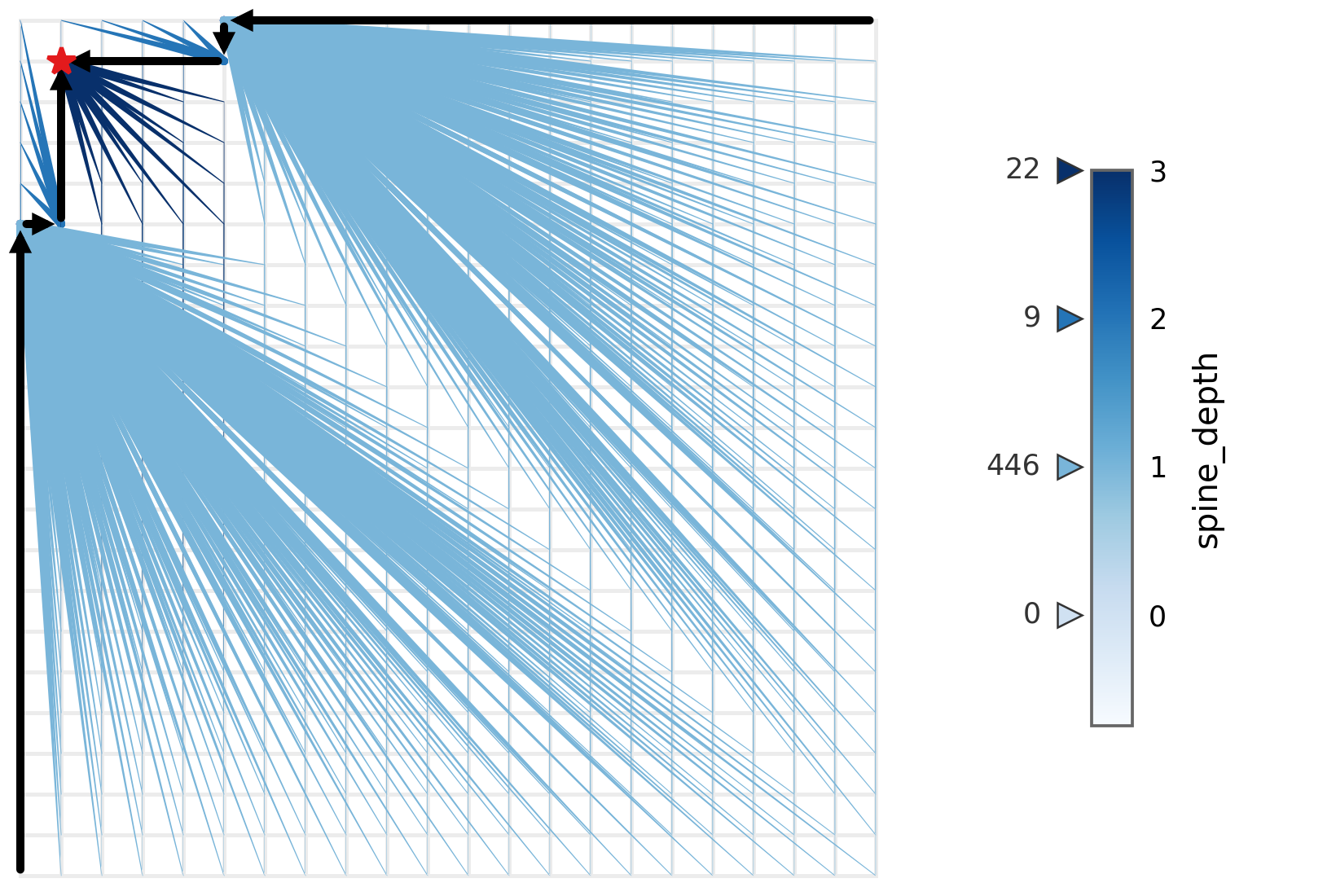}
        \caption{Visualization of a function on $[22]^2$.}
        \label{fig:function-21-2}
    \end{figure}

    Each number to the left of the color scale counts the number of the inputs $\boldsymbol{x}$ for which the distance from $\boldsymbol{f}(\boldsymbol{x})$ to the endpoint of the spine equals the value shown to the right of the color scale.
\end{example}

To illustrate the geometric structure of our hard instance, we visualize a hard instance on $[2]^8$ in Example \ref{eg:function-2-8}.

\begin{example} \label{eg:function-2-8}
    \normalfont
    Consider the token permutation 
    $$
    \pi = [(0, 0), (3, 0), (2, 0), (7, 0), (5, 0), (6, 0), (4, 0), (1, 0)].
    $$

    The function is visualized in the same manner as described in Example \ref{eg:function-21-2}, and the resulting visualization is shown in Figure \ref{fig:function-2-8}.

    \begin{figure}[H]
        \centering
        \includegraphics[width=0.5\linewidth]{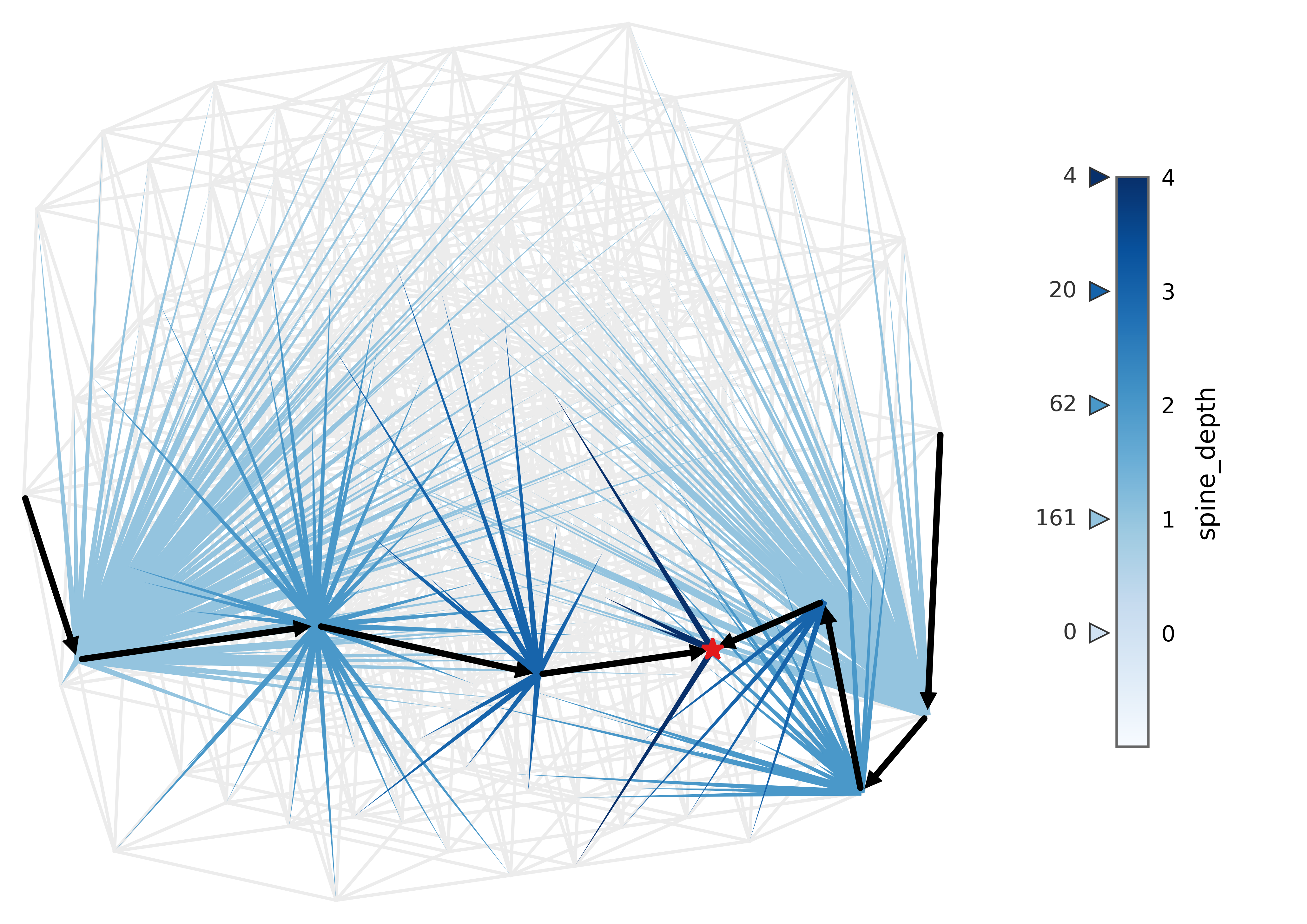}
        \caption{Visualization of a function on $[2]^8$.}
        \label{fig:function-2-8}
    \end{figure}

    The statistics on the right suggest that the number of points projected to successive depths of the spine decreases geometrically. This indicates that most points reveal only a small amount of information about the spine, consistent with the intuition that its structure becomes increasingly difficult to distinguish closer to the fixed point.
\end{example}

We now verify that the construction gives valid Tarski instances with a unique
fixed point.

\begin{lemma} \label{lem:unique_fixed_point}
For every $\pi\in\textsf{P}_{\textsf{T}}$, the unique fixed point of $\boldsymbol{f}_\pi$ is
$\boldsymbol{x}_m^\pi$.
\end{lemma}

\begin{proof}
Since the image of $\boldsymbol{f}_\pi$ is contained in $\mathcal{X}_\pi$, every fixed
point of $\boldsymbol{f}_\pi$ must lie on the spine. For every $i\in[N+1]$, $\boldsymbol{f}_\pi(\boldsymbol{x}_i^\pi)=\boldsymbol{x}_{h(i)}^\pi$. This equals $\boldsymbol{x}_i^\pi$ if and only if $i=m$. Therefore,
$\boldsymbol{x}_m^\pi$ is the unique fixed point of $\boldsymbol{f}_\pi$.
\end{proof}

\begin{lemma}\label{lem:monotonicity}
For every $\pi\in\textsf{P}_{\textsf{T}}$, the function $\boldsymbol{f}_\pi$ is monotone.
\end{lemma}

\begin{proof}
Let $\boldsymbol{u},\boldsymbol{v}\in[n]^k$ satisfy
$\boldsymbol{u}\preceq\boldsymbol{v}$. Then
$\norm{\boldsymbol{u}}\leq\norm{\boldsymbol{v}}$. We first show that $\rho_\pi(\boldsymbol{u})\leq\rho_\pi(\boldsymbol{v})$.

If both points lie in the lower half of the lattice, every spine vertex below
$\boldsymbol{u}$ is also below $\boldsymbol{v}$, so the inequality follows from
the two maxima defining $\rho_\pi$. If both lie in the upper half, every spine
vertex above $\boldsymbol{v}$ is also above $\boldsymbol{u}$, so the inequality
follows from the two minima. In the remaining case,
$\norm{\boldsymbol{u}}\leq W/2<\norm{\boldsymbol{v}}$, and $\boldsymbol{x}_{\rho_\pi(\boldsymbol{u})}^\pi
\preceq\boldsymbol{u}\preceq\boldsymbol{v}
\preceq\boldsymbol{x}_{\rho_\pi(\boldsymbol{v})}^\pi$, which gives the same inequality because the spine is strictly increasing.

The map $h$ is nondecreasing. Therefore, $h(\rho_\pi(\boldsymbol{u}))
\leq h(\rho_\pi(\boldsymbol{v})),$ and the ordering of the spine yields
$\boldsymbol{f}_\pi(\boldsymbol{u})\preceq \boldsymbol{f}_\pi(\boldsymbol{v})$.
\end{proof}

\section{Tree-Filtration Adversary Method} \label{sec:adversary_method}

\subsection{Introduction to the General Framework} \label{subsec:general_framework}
We introduce the Tree-Filtration Adversary Method in this section. From a high-level perspective, our method requires:

\begin{enumerate}
    \item All hard instances are organized as an equal-depth (Definition \ref{def:equal-depth-tree}) \emph{filtration tree} $\mathcal{T}$ (Definition \ref{def:decision-tree}), with all leaves having depth $L$. Every leaf represents a hard instance. Every non-leaf node represents a computationally indistinguishable subset of hard instances. This tree structure should provide a \emph{simplified difference criterion} $\hat{D}(u)$ on each node (Definition \ref{def:simplified_difference_criterion}) for determining whether $\boldsymbol{x}^1_i$ and $\boldsymbol{x}^2_i$ differ for $\boldsymbol{x}^1, \boldsymbol{x}^2\in \mathcal{F}(u)$ and $i\in I$. 
    \item There exists a \emph{distinguishing random walk} $\mathcal{A}$ (Definition \ref{def:distinguishing-random-walk}), such that:
    \begin{enumerate}
        \item The distribution of the $f$-values in any subtree is anti-concentrated (Definition \ref{def:anti-concentration}). Intuitively, the output in the subtree rooted at $u$ remains dispersed even after the information represented by $u$ is revealed.
        \item For any fixed $i\in I$, the probability that $\mathcal{A}$ reaches a node $v$ with $i\in \hat{D}(v)$ decays exponentially with respect to the depth (Definition \ref{def:geometric-indistinguishable}). Intuitively, the distinguishing power of any fixed query decays geometrically as the walk descends the tree.
    \end{enumerate}
\end{enumerate}

Then, the adversary matrix and an adversary bound can be derived from our construction in Definition \ref{def:adversary-matrix}. Next, we will formalize the above intuitions step by step. 

\begin{definition}[Equal depth tree] \label{def:equal-depth-tree}
A rooted tree $\mathcal{T}$ is \emph{$L$-equal-depth} if
$\mathrm{dep}(\ell)=L$ for every $\ell\in\mathcal{L}(\mathcal{T})$.
\end{definition}

\begin{definition}[Hard-instance filtration tree] \label{def:decision-tree}
Let $S\subseteq\Sigma^I$, and let $f\colon S\to\Pi$ be a possibly
partial function. For a family $\mathcal{F}\subseteq S$ of hard instances, an $L$-equal-depth rooted tree $\mathcal{T}$ is a \emph{hard-instance filtration tree} of $\mathcal{F}$ if there is a bijection $\mathrm{ins}\colon\mathcal{L}(\mathcal{T})\to\mathcal{F}$. Thus,
$\mathrm{ins}(\ell)$ is the instance represented by the leaf $\ell$.
\end{definition}

Throughout the remainder of this subsection, fix finite alphabets $\Sigma$ and
$\Pi$, an index set $I$ with cardinality $n$, a promise set $S\subseteq\Sigma^I$, a function $f\colon S\to\Pi$, a
hard-instance family $\mathcal{F}\subseteq S$, and an $L$-equal-depth
hard-instance filtration tree
$\mathcal{T}=(\mathcal{V},\mathcal{E},r_{\mathcal{T}})$ for $\mathcal{F}$.
For $u\in\mathcal{V}$, write
$\mathcal{F}(u):=\{\mathrm{ins}(\ell):
\ell\in\mathcal{L}(\mathrm{Sub}(u))\}$ for the instances represented below
$u$.

The tree records how instances are grouped. To connect this hierarchy with
queries, we first record the coordinates on which two instances differ.

\begin{definition}[Difference criterion]
The \emph{difference criterion} of $f$ is the map
$D\colon S\times S\to\mathcal{P}(I)$ defined by
$D(x,y):=\{i\in I:x_i\ne y_i\}$.
\end{definition}

The exact set $D(x,y)$ depends on a pair of leaves. The following node-level
upper bound is easier to use and is required to shrink along the filtration.

\begin{definition}[Simplified difference criterion] \label{def:simplified_difference_criterion}
A map $\hat{D}\colon\mathcal{V}\to\mathcal{P}(I)$ is a
\emph{simplified difference criterion} on $\mathcal{T}$ if
\begin{enumerate}[label=(\roman*)]
    \item $D(x,y)\subseteq\hat{D}(u)$ for every $u\in\mathcal{V}$ and all $x,y\in\mathcal{F}(u)$; and
    \item $\hat{D}$ is \emph{nested} along $T$, meaning that $\hat{D}(v)\subseteq\hat{D}(u)$ whenever $v$ is a descendant of $u$.
\end{enumerate}

\end{definition}

Thus, $\hat{D}(u)$ contains every coordinate on which two instances below
$u$ may differ, while the possible difference sets can only decrease as one
moves from the root toward a leaf.

We next equip the filtration tree with a probability distribution on each
level.

\begin{definition}[Distinguishing random walk] \label{def:distinguishing-random-walk}
A \emph {distinguishing random walk} $\mathcal{A}$ is a
top-down random walk on $\mathcal{T}$. We write
$\mathcal{A}\xrightarrow{t}u$ for the event that the walk reaches $u$ at
depth $t$. For every non-leaf node $u$ and every $s\in\mathrm{ch}(u)$, let
\[
P_{\mathcal{A}}(s\mid u)
:=
\Pr[\mathcal{A}\xrightarrow{\mathrm{dep}(u)+1}s
\mid \mathcal{A}\xrightarrow{\mathrm{dep}(u)}u].
\]
The transition probabilities are required to satisfy
\begin{equation}
P_{\mathcal{A}}(s\mid u)>0,
\qquad
\sum_{s\in\mathrm{ch}(u)}P_{\mathcal{A}}(s\mid u)=1,
\qquad
\Pr[\mathcal{A}\xrightarrow{0}r_{\mathcal{T}}]=1.
\label{eq:pda-sum-1}
\end{equation}
\end{definition}

Strict positivity makes every node reachable, so all conditional
probabilities used below are well defined.

\begin{definition}[Induced distribution] \label{def:induced-distribution}
Let $r_{\mathcal{T}}=u_0\to u_1\to\cdots\to u_d=u$ be the unique path from
the root to $u$. The probability induced by $\mathcal{A}$ is
\begin{equation}
\Pr_{\mathcal{A}}[u]
:=
\prod_{j=1}^{d}P_{\mathcal{A}}(u_j\mid u_{j-1}),
\qquad
\Pr_{\mathcal{A}}[r_{\mathcal{T}}]:=1.
\label{eq:induced-distribution}
\end{equation}
\end{definition}

The mass of a node is precisely the total mass of the leaves below it.

\begin{lemma}\label{lem:subtree-mass}
For every $u\in\mathcal{V}$,
\[
\Pr_{\mathcal{A}}[u]
=
\sum_{\ell\in\mathcal{L}(\mathrm{Sub}(u))}
\Pr_{\mathcal{A}}[\ell].
\]
\end{lemma}

\begin{proof}
We argue by induction on $L-\mathrm{dep}(u)$. The claim is immediate when
$u$ is a leaf. Otherwise, the induction hypothesis and
Equation~\eqref{eq:pda-sum-1} give
\[
\begin{aligned}
\sum_{\ell\in\mathcal{L}(\mathrm{Sub}(u))}
\Pr_{\mathcal{A}}[\ell]
&=
\sum_{s\in\mathrm{ch}(u)}
\sum_{\ell\in\mathcal{L}(\mathrm{Sub}(s))}
\Pr_{\mathcal{A}}[\ell]
\quad\text{(grouping the leaves by child subtree)}\\
&=
\sum_{s\in\mathrm{ch}(u)}\Pr_{\mathcal{A}}[s]
\quad\text{(By induction hypothesis)}\\
&=
\Pr_{\mathcal{A}}[u]
\sum_{s\in\mathrm{ch}(u)}P_{\mathcal{A}}(s\mid u)
\quad\text{(By Equation~\eqref{eq:pda-sum-1})} \\
&=\Pr_{\mathcal{A}}[u].
\end{aligned}
\]
\end{proof}

Accordingly, conditioned on reaching $u$, every descendant $v$ at depth
$s\geq\mathrm{dep}(u)$ is sampled with probability
$\Pr_{\mathcal{A}}[v]/\Pr_{\mathcal{A}}[u]$. We write
$V\xleftarrow{s}\mathcal{A}$ for the resulting random node $V$ at depth $s$. 

We need two complementary properties of this random walk. The first ensures
that the target value remains sufficiently dispersed inside every subtree;
the second controls how quickly a fixed query coordinate disappears from the
simplified difference criterion.

\begin{definition}[Anti-concentration property] \label{def:anti-concentration}
Let $\boldsymbol{a}=(a_0,\ldots,a_{L - 1})$ be a nonnegative sequence. We say that $\mathcal{A}$ satisfies the
\emph{$\boldsymbol{a}$-anti-concentration property} if, for every node $u$
at depth $t<L$ and every $\ell\in\mathcal{L}(\mathrm{Sub}(u))$,
\[
\Pr_{V\xleftarrow{L}\mathcal{A}}
\bigl[f(\mathrm{ins}(\ell))\ne f(\mathrm{ins}(V))
\mid \mathcal{A}\xrightarrow{t}u\bigr]
\geq a_t.
\]
\end{definition}

\begin{definition}[Geometric-indistinguishability property] \label{def:geometric-indistinguishable}
For $0<r<1$, we say that $\mathcal{A}$ satisfies the
\emph{$r$-geometric-indistinguishability property} if, for every
$i\in I$, every pair of depths $0\leq t<s\leq L$, and every node $u$ at
depth $t$,
\[
\Pr_{V\xleftarrow{s}\mathcal{A}}
\bigl[i\in\hat{D}(V)\mid\mathcal{A}\xrightarrow{t}u\bigr]
\leq r^{s-t}.
\]
\end{definition}

Figure~\ref{fig:properties} illustrates the two properties. Anti-concentration
will lower bound the numerator of the adversary ratio, while geometric-indistinguishability will upper bound its denominator.

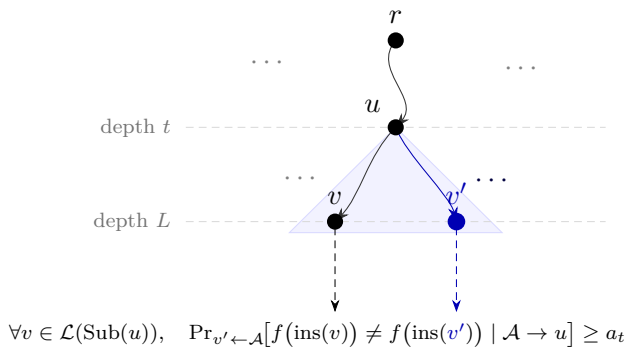
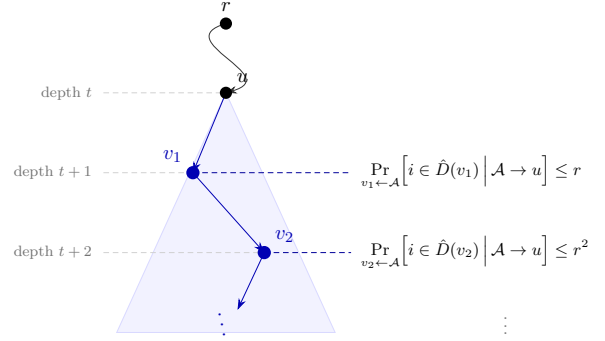
\begin{figure}[H]
    \centering
    \begin{subfigure}{0.48\textwidth}
        \centering
        \resizebox{\linewidth}{!}{%
\begin{tikzpicture}[
    treepath/.style={
        thin,
        draw=black!80,
        -{Stealth[length=2mm,width=1.35mm]}
    },
    bluepath/.style={
        thin,
        draw=blue!70!black,
        -{Stealth[length=2mm,width=1.35mm]}
    },
    guidearrow/.style={
        thin,
        densely dashed,
        -{Stealth[length=1.7mm,width=1.2mm]}
    },
    blacknode/.style={
        circle,
        fill=black,
        inner sep=2.2pt
    },
    bluenode/.style={
        circle,
        fill=blue!70!black,
        inner sep=2.4pt
    },
    depthline/.style={
        densely dashed,
        draw=black!15
    },
    depthlabel/.style={
        font=\scriptsize,
        text=black!55
    }
]

% ------------------------------------------------------------------
% Formula
% ------------------------------------------------------------------

\matrix (formula) [
    matrix of math nodes,
    ampersand replacement=\&,
    nodes={
        font=\scriptsize,
        inner sep=0pt,
        outer sep=0pt,
        text height=1.8ex,
        text depth=0.6ex
    },
    column sep=0pt,
    anchor=center
] at (0,0) {
    \forall v\in\mathcal{L}(\mathrm{Sub}(u)),\quad
    \Pr_{v'\leftarrow\mathcal{A}}\!
    \bigl[f\bigl(\mathrm{ins}(
    \&
    {\color{black}v}
    \&
    )\bigr)\ne f\bigl(\mathrm{ins}(
    \&
    {\color{blue!70!black}v'}
    \&
    )\bigr)
    \mid \mathcal{A}\rightarrow u
    \bigr]\geq a_t
    \\
};

\coordinate (formulaV)  at (formula-1-2.north);
\coordinate (formulaVp) at (formula-1-4.north);

% ------------------------------------------------------------------
% Tree coordinates
% ------------------------------------------------------------------

\coordinate (v)  at ($(formulaV)+(0,1.30)$);
\coordinate (vp) at ($(formulaVp)+(0,1.30)$);
\coordinate (leafmid) at ($(v)!0.5!(vp)$);

\coordinate (u) at ($(leafmid)+(0,1.28)$);
\coordinate (r) at ($(u)+(0,1.18)$);

% ------------------------------------------------------------------
% Control points
%
% Each pair is obtained from points on the corresponding straight
% segment. The first control point is shifted left; the second is
% shifted right. The small shifts keep both curves inside Sub(u).
% ------------------------------------------------------------------

% Control points for r -> u.
\coordinate (cruA)
    at ($(r)!0.33!(u)+(-0.34,0)$);
\coordinate (cruB)
    at ($(r)!0.67!(u)+( 0.34,0)$);

% Control points for u -> v.
\coordinate (cuvA)
    at ($(u)!0.30!(v)+(-0.14,0)$);
\coordinate (cuvB)
    at ($(u)!0.72!(v)+( 0.12,0)$);

% Control points for u -> v'.
\coordinate (cuvpA)
    at ($(u)!0.30!(vp)+(-0.10,0)$);
\coordinate (cuvpB)
    at ($(u)!0.72!(vp)+( 0.14,0)$);

% ------------------------------------------------------------------
% Light blue triangle representing Sub(u)
% ------------------------------------------------------------------

\filldraw[
    fill=blue!5,
    draw=blue!15,
    thin
]
    (u)
    --
    ($(v)+(-0.62,-0.15)$)
    --
    ($(vp)+(0.62,-0.15)$)
    -- cycle;

% ------------------------------------------------------------------
% Depth guides
% ------------------------------------------------------------------

\draw[depthline]
    ($(u)+(-2.85,0)$)
    --
    ($(u)+(2.85,0)$);

\draw[depthline]
    ($(leafmid)+(-2.85,0)$)
    --
    ($(leafmid)+(2.85,0)$);

\node[
    depthlabel,
    anchor=east
] at ($(u)+(-2.92,0)$)
{depth \(t\)};

\node[
    depthlabel,
    anchor=east
] at ($(leafmid)+(-2.92,0)$)
{depth \(L\)};

% ------------------------------------------------------------------
% S-shaped paths
% ------------------------------------------------------------------

% r -> u: first left, then right.
\draw[treepath]
    (r)
    .. controls
        (cruA)
        and
        (cruB)
    ..
    (u);

% u -> v:
% contained in the left portion of the triangle.
\draw[treepath]
    (u)
    .. controls
        (cuvA)
        and
        (cuvB)
    ..
    (v);

% u -> v':
% contained in the right portion of the triangle.
\draw[bluepath]
    (u)
    .. controls
        (cuvpA)
        and
        (cuvpB)
    ..
    (vp);

% ------------------------------------------------------------------
% Other unspecified parts of the tree
% ------------------------------------------------------------------

\node[text=black!50]
    at ($(r)+(-1.75,-0.30)$)
{\(\cdots\)};

\node[text=black!50]
    at ($(r)+(1.70,-0.38)$)
{\(\cdots\)};

\node[text=black!50]
    at ($(u)+(-1.30,-0.70)$)
{\(\cdots\)};

\node[text=blue!25!black]
    at ($(u)+(1.30,-0.72)$)
{\(\cdots\)};

% ------------------------------------------------------------------
% Distinguished vertices
% ------------------------------------------------------------------

\node[blacknode] at (r) {};
\node[above=3pt] at (r) {\(r\)};

\node[blacknode] at (u) {};
\node[above left=2pt] at (u) {\(u\)};

\node[blacknode] at (v) {};
\node[above=3pt] at (v) {\(v\)};

\node[bluenode] at (vp) {};
\node[
    text=blue!70!black,
    above=3pt
] at (vp)
{\(v'\)};

% ------------------------------------------------------------------
% Short dashed arrows to the formula
% ------------------------------------------------------------------

\draw[
    guidearrow,
    draw=black
]
    ($(v)+(0,-0.08)$)
    --
    ($(formulaV)+(0,0.07)$);

\draw[
    guidearrow,
    draw=blue!70!black
]
    ($(vp)+(0,-0.08)$)
    --
    ($(formulaVp)+(0,0.07)$);

\end{tikzpicture}%
}
        \caption{Visualization of Definition \ref{def:anti-concentration}.}
        \label{fig:anti-concentration}
    \end{subfigure}
    \hfill
    \begin{subfigure}{0.48\textwidth}
        \centering
        \makebox[\linewidth][c]{%
\resizebox{0.92\linewidth}{!}{%
\begin{tikzpicture}[
    blackpath/.style={
        thin,
        draw=black!80,
        -{Stealth[length=2mm,width=1.35mm]}
    },
    bluepath/.style={
        thin,
        draw=blue!70!black,
        -{Stealth[length=2mm,width=1.35mm]}
    },
    blacknode/.style={
        circle,
        fill=black,
        inner sep=2.3pt
    },
    bluenode/.style={
        circle,
        fill=blue!70!black,
        inner sep=2.5pt
    },
    depthline/.style={
        densely dashed,
        thin,
        draw=black!18
    },
    probabilityline/.style={
        densely dashed,
        thin,
        draw=blue!60!black
    },
    depthlabel/.style={
        font=\scriptsize,
        text=black!55
    },
    formulalabel/.style={
        font=\footnotesize,
        text=black,
        anchor=west,
        inner sep=1.5pt,
        fill=white
    }
]

% Vertical padding only
\path[draw=none]
    (0,-0.55) -- (0,6.25);

% Main coordinates
\coordinate (r)  at ( 0.00,5.75);
\coordinate (u)  at ( 0.00,4.45);
\coordinate (v1) at (-0.62,2.95);
\coordinate (v2) at ( 0.72,1.45);
\coordinate (vnext) at (0.22,0.38);

\coordinate (depthT)  at (-2.30,4.45);
\coordinate (depthT1) at (-2.30,2.95);
\coordinate (depthT2) at (-2.30,1.45);

\coordinate (prob1) at (2.35,2.95);
\coordinate (prob2) at (2.35,1.45);

% Subtree triangle
\filldraw[
    fill=blue!5,
    draw=blue!15,
    thin
]
    (u)
    --
    (-2.05,-0.05)
    --
    (2.05,-0.05)
    -- cycle;

% Depth guides
\draw[depthline]
    (depthT) -- (u);

\draw[depthline]
    (depthT1) -- (v1);

\draw[depthline]
    (depthT2) -- (v2);

\node[
    depthlabel,
    anchor=east
] at ($(depthT)+(-0.08,0)$)
{depth \(t\)};

\node[
    depthlabel,
    anchor=east
] at ($(depthT1)+(-0.08,0)$)
{depth \(t+1\)};

\node[
    depthlabel,
    anchor=east
] at ($(depthT2)+(-0.08,0)$)
{depth \(t+2\)};

% Probability guides
\draw[probabilityline]
    (v1) -- (prob1);

\draw[probabilityline]
    (v2) -- (prob2);

% Probability formulas
\node[
    formulalabel
] at ($(prob1)+(0.12,0)$)
{%
    \(\displaystyle
    \Pr_{v_1\leftarrow\mathcal{A}}
    \!\left[
        i\in\hat{D}(v_1)
        \,\middle|\,
        \mathcal{A}\rightarrow u
    \right]
    \leq r
    \)
};

\node[
    formulalabel
] at ($(prob2)+(0.12,0)$)
{%
    \(\displaystyle
    \Pr_{v_2\leftarrow\mathcal{A}}
    \!\left[
        i\in\hat{D}(v_2)
        \,\middle|\,
        \mathcal{A}\rightarrow u
    \right]
    \leq r^2
    \)
};

\node[
    font=\footnotesize,
    text=black!55
] at (5.25,0.18)
{\(\vdots\)};

% r -> u:
% the curve first moves left and then returns from the right.
\draw[blackpath]
    (r)
    .. controls
        ($(r)+(-1.10,-0.32)$)
        and
        ($(u)+(1.10,0.32)$)
    ..
    (u);

% Straight blue path
\draw[bluepath]
    (u) -- (v1);

\draw[bluepath]
    (v1) -- (v2);

\draw[bluepath]
    (v2) -- (vnext);

% Distinguished vertices
\node[blacknode] at (r) {};
\node[above=3pt] at (r) {\(r\)};

\node[blacknode] at (u) {};
\node[above right=2pt] at (u) {\(u\)};

\node[bluenode] at (v1) {};
\node[
    text=blue!70!black,
    above left=2pt
] at (v1)
{\(v_1\)};

\node[bluenode] at (v2) {};
\node[
    text=blue!70!black,
    above right=2pt
] at (v2)
{\(v_2\)};

% Continuation ellipsis
\node[
    text=blue!70!black,
    font=\large,
    rotate=-45
] at (0.00,0.18)
{\(\ddots\)};

\end{tikzpicture}%
}%
}
        \caption{Visualization of Definition \ref{def:geometric-indistinguishable}.}
        \label{fig:geometric-indistinguishability}
    \end{subfigure}

    \caption{{Visualization of the two properties.} In these figures, black nodes / black arrows represent fixed nodes / paths, while blue nodes / blue arrows represent random nodes / paths sampled by $\mathcal{A}$. Light blue triangles represent subtrees.}
    \label{fig:properties}
\end{figure}

Now, we can state the tree-filtration adversary method as follows:

\begin{theorem}[Tree-filtration adversary method] \label{thm:tree-adversary}
Suppose that the hard-instance family $\mathcal{F}$ admits an
$L$-equal-depth hard-instance filtration tree $\mathcal{T}$ equipped with a
simplified difference criterion $\hat{D}$. If a distinguishing random walk
$\mathcal{A}$ on $\mathcal{T}$ satisfies the
$\boldsymbol{a}$-anti-concentration property and the
$r$-geometric-indistinguishability property, then, for
$0\leq\epsilon<1/2$,
\[
Q_\epsilon(f)
\geq
\left(1-2\sqrt{\epsilon(1-\epsilon)}\right)
\frac{1-\sqrt{r}}{1+\sqrt{r}}
\sum_{t=0}^{L - 1}a_t.
\]
\end{theorem}

We prove the theorem by constructing an adversary matrix for the restriction
$f|_{\mathcal{F}}$. The entries of the adversary matrix are labeled with $(x, y)\in \mathcal{F}\times \mathcal{F}$. For $x\in\mathcal{F}$, let
$\ell_x:=\mathrm{ins}^{-1}(x)$ and
$p_x:=\Pr_{\mathcal{A}}[\ell_x]$; for a node $u$, let
$p_u:=\Pr_{\mathcal{A}}[u]$.

\begin{definition}[Construction of the adversary matrix] \label{def:adversary-matrix}
For every $u\in\mathcal{V}$, define the normalized subtree vector
$\boldsymbol{\phi}_u\in\mathbb{R}^{\mathcal{F}}$ and the matrix
$\boldsymbol{\Gamma}_u$ by
\[
\boldsymbol{\phi}_u[x]
:=
\mathbf{1}\{x\in\mathcal{F}(u)\}\sqrt{\frac{p_x}{p_u}},
\qquad
\boldsymbol{\Gamma}_u[x,y]
:=
\boldsymbol{\phi}_u[x]\boldsymbol{\phi}_u[y]
\mathbf{1}\{f(x)\ne f(y)\}.
\]
For $0\leq t< L$, let
$\boldsymbol{\Gamma}_t:=
\sum_{u\in\mathcal{V}:\,\mathrm{dep}(u)=t}\boldsymbol{\Gamma}_u$, and set
$\boldsymbol{\Gamma}:=\sum_{t=0}^{L - 1}\boldsymbol{\Gamma}_t$.
\end{definition}

The matrix $\boldsymbol{\Gamma}$ is nonnegative and symmetric, and its
$(x,y)$ entry vanishes whenever $f(x)=f(y)$. It is therefore a valid
nonnegative adversary matrix for $f|_{\mathcal{F}}$. We will bound the numerator of the adversary bound in Lemma \ref{lem:numerator-lb}, and the denominator in Lemma \ref{lem:denominator-ub}.

\begin{lemma}[Lower bound of the numerator] \label{lem:numerator-lb}
    For the $\boldsymbol{\Gamma}$ constructed in Definition \ref{def:adversary-matrix}, we have
    $$
    \norm{\boldsymbol{\Gamma}} \geq \sum_{i = 0}^{L - 1} a_i.
    $$
\end{lemma}

\begin{proof}
Let $\boldsymbol{v}\in\mathbb{R}^{\mathcal{F}}$ be given by
$\boldsymbol{v}[x]=\sqrt{p_x}$. Since
$\Pr_{\mathcal{A}}[r_{\mathcal{T}}]=1$, Lemma~\ref{lem:subtree-mass} gives
$\norm{\boldsymbol{v}}_2^2=\sum_{x\in\mathcal{F}}p_x=1$. Thus,
$\boldsymbol{v}$ is a unit vector. For every non-leaf node $u$ and every
$x\in\mathcal{F}(u)$,
\[
\begin{aligned}
(\boldsymbol{\Gamma}_u\boldsymbol{v})[x]
&=
\sum_{y\in\mathcal{F}}
\boldsymbol{\Gamma}_u[x,y]\boldsymbol{v}[y] \\
&=
\sum_{y\in\mathcal{F}}
\boldsymbol{\phi}_u[x]\boldsymbol{\phi}_u[y]
\mathbf{1}\{f(x)\ne f(y)\}\sqrt{p_y} \\
&=
\frac{\sqrt{p_x}}{p_u}
\sum_{y\in\mathcal{F}(u)}p_y\mathbf{1}\{f(x)\ne f(y)\} \\
&=
\sqrt{p_x}\,
\Pr_{Y\xleftarrow{L}\mathcal{A}}
\bigl[f(x)\ne f(\mathrm{ins}(Y))
\mid\mathcal{A}\xrightarrow{\mathrm{dep}(u)}u\bigr] \\
&\geq
a_{\mathrm{dep}(u)}\sqrt{p_x},
\end{aligned}
\]
where the inequality follows from anti-concentration. For every $t<L$, the
node $u_t(x):=\mathrm{anc}(\ell_x,t)$ is the unique node at depth $t$ for
which $x\in\mathcal{F}(u_t(x))$. Every other node $u$ at that depth satisfies
$(\boldsymbol{\Gamma}_u\boldsymbol{v})[x]=0$. Therefore,
$(\boldsymbol{\Gamma}_t\boldsymbol{v})[x]
=(\boldsymbol{\Gamma}_{u_t(x)}\boldsymbol{v})[x]
\geq a_t\boldsymbol{v}[x]$. Consequently,
entrywise,
\[
\boldsymbol{\Gamma}\boldsymbol{v}
\geq
\left(\sum_{t=0}^{L - 1}a_t\right)\boldsymbol{v}.
\]
All vectors in this inequality are nonnegative, so
$\norm{\boldsymbol{\Gamma}}\geq
\norm{\boldsymbol{\Gamma}\boldsymbol{v}}_2
\geq\sum_{t=0}^{L - 1}a_t$.
\end{proof}

We now turn to the denominator. Fix $i\in I$ and, for every node $u$,
define 
\[
\mathbf{Q}_u
:=
\mathbf{1}\{i\in\hat{D}(u)\}
\boldsymbol{\phi}_u\boldsymbol{\phi}_u^{\mathsf{T}},
\qquad
\mathbf{Q}_t
:=
\sum_{u:\,\mathrm{dep}(u)=t}\mathbf{Q}_u,
\qquad
\mathbf{Q}
:=
\sum_{t=0}^{L - 1}\mathbf{Q}_t.
\]
If $(\boldsymbol{\Gamma}_u\circ\boldsymbol{\Delta}_i)[x,y]>0$,
then $x,y\in\mathcal{F}(u)$ and $i\in D(x,y)$. The simplified difference
criterion therefore implies $i\in\hat{D}(u)$. Thus,
$(\boldsymbol{\Gamma}_u\circ\boldsymbol{\Delta}_i)[x,y]
=\boldsymbol{\phi}_u[x]\boldsymbol{\phi}_u[y]
\mathbf{1}\{f(x)\ne f(y)\}\mathbf{1}\{x_i\ne y_i\}
\leq \boldsymbol{\phi}_u[x]\boldsymbol{\phi}_u[y] = 
\mathbf{Q}_u[x,y].$ Therefore, for every node $u$, $\boldsymbol{0}\leq\boldsymbol{\Gamma}_u\circ\boldsymbol{\Delta}_i\leq\mathbf{Q}_u$ entrywise. Since $\boldsymbol{\Gamma}=\sum_{u\in\mathcal{V} \setminus \mathcal{L}(\mathcal{T})}\boldsymbol{\Gamma}_u$ and
$\mathbf{Q}=\sum_{u\in\mathcal{V} \setminus \mathcal{L}(\mathcal{T})}\mathbf{Q}_u$, summing the preceding
inequality over $u$ gives
\[
\boldsymbol{0}
\leq
\boldsymbol{\Gamma}\circ\boldsymbol{\Delta}_i
=
\sum_{u\in\mathcal{V}\setminus \mathcal{L}(\mathcal{T})}
\boldsymbol{\Gamma}_u\circ\boldsymbol{\Delta}_i
\leq
\sum_{u\in\mathcal{V}\setminus \mathcal{L}(\mathcal{T})}\mathbf{Q}_u
=
\mathbf{Q}.
\]
By Lemma~\ref{lem:spectral-norm-monotonicity},
$\norm{\boldsymbol{\Gamma}\circ\boldsymbol{\Delta}_i}
\leq\norm{\mathbf{Q}}$. We only need to upper bound $\norm{\mathbf{Q}}$.

\begin{lemma} \label{lem:idempotent}
For every $0\leq t< L$, the matrix $\mathbf{Q}_t$ is real symmetric and
idempotent.
\end{lemma}

\begin{proof}
Lemma~\ref{lem:subtree-mass} gives
$\norm{\boldsymbol{\phi}_u}_2=1$. Moreover, distinct nodes at the same
depth have disjoint sets of descendant leaves, so their subtree vectors are
orthogonal. Thus $\mathbf{Q}_t$ is the sum of a collection of mutually
orthogonal rank-one projections. It is therefore real symmetric and
idempotent.
\end{proof}

\begin{lemma} \label{lem:geometric-supression}
    For any $s, t\in \{0, \ldots , L - 1\}$, 
    $$
    \norm{\mathbf{Q}_s\mathbf{Q}_t} \leq r^{|s - t| / 2}.
    $$
\end{lemma}
\begin{proof}
By Lemma~\ref{lem:idempotent}, each $\mathbf{Q}_t$ is an orthogonal
projection. Hence, if $s=t$, $\norm{\mathbf{Q}_s\mathbf{Q}_t}=\norm{\mathbf{Q}_t^2}
=\norm{\mathbf{Q}_t}\leq 1=r^{|s-t|/2}.$

Suppose now that $s>t$. For $q\in\{t,s\}$, write $\mathcal{U}_{q,i}:=\{u\in\mathcal{V}:\mathrm{dep}(u)=q,\ i\in\hat{D}(u)\}.$ For $u\in\mathcal{U}_{t,i}$, let $\mathcal{W}_s(u):=\{w\in\mathcal{U}_{s,i}:w\text{ is a descendant of }u\}$.

We first compute the inner product between subtree vectors at the two
depths. Fix $u\in\mathcal{U}_{t,i}$ and $w\in\mathcal{U}_{s,i}$. If $w$ is a descendant of $u$, then
$\mathcal{F}(w)\subseteq\mathcal{F}(u)$, and therefore
\begin{align*}
\left\langle\boldsymbol{\phi}_u,
\boldsymbol{\phi}_w\right\rangle
&=
\sum_{x\in\mathcal{F}}
\mathbf{1}\{x\in\mathcal{F}(u)\}
\mathbf{1}\{x\in\mathcal{F}(w)\}
\sqrt{\frac{p_x}{p_u}}
\sqrt{\frac{p_x}{p_w}} \\
&=
\frac{1}{\sqrt{p_up_w}}
\sum_{x\in\mathcal{F}(w)}p_x\qquad\text{(By $\mathcal{F}(w)\subseteq\mathcal{F}(u)$)}\\
&=
\frac{p_w}{\sqrt{p_up_w}}\qquad\text{(By Lemma~\ref{lem:subtree-mass})}\\
&=
\sqrt{\frac{p_w}{p_u}}.
\end{align*}
If $w$ is not a descendant of $u$, then $\mathcal{F}(u)\cap\mathcal{F}(w)=\varnothing$. It follows directly from the definition of the subtree vectors that
$\langle\boldsymbol{\phi}_u,\boldsymbol{\phi}_w\rangle=0$.

By definition,
$\mathbf{Q}_q=\sum_{v\in\mathcal{U}_{q,i}}
\boldsymbol{\phi}_v\boldsymbol{\phi}_v^{\mathsf{T}}$ for
$q\in\{t,s\}$. Hence
\begin{align*}
\mathbf{Q}_t\mathbf{Q}_s\mathbf{Q}_t
&=
\sum_{\substack{u,v\in\mathcal{U}_{t,i}\\
                 w\in\mathcal{U}_{s,i}}}
\boldsymbol{\phi}_u
\left\langle\boldsymbol{\phi}_u,
\boldsymbol{\phi}_w\right\rangle
\left\langle\boldsymbol{\phi}_w,
\boldsymbol{\phi}_v\right\rangle
\boldsymbol{\phi}_v^{\mathsf{T}} \\
&=
\sum_{u\in\mathcal{U}_{t,i}}
\sum_{w\in\mathcal{W}_s(u)}
\frac{p_w}{p_u}
\boldsymbol{\phi}_u\boldsymbol{\phi}_u^{\mathsf{T}} \\
&=
\sum_{u\in\mathcal{U}_{t,i}}
\alpha_u
\boldsymbol{\phi}_u\boldsymbol{\phi}_u^{\mathsf{T}},
\end{align*}
where $\alpha_u:=\sum_{w\in\mathcal{W}_s(u)}p_w/p_u$. To justify the
second equality, observe that $w$ has a unique ancestor at depth $t$.
Thus a term in the first sum can be nonzero only when $u=v$ is that
ancestor. Moreover, the nesting of $\hat{D}$ guarantees that the depth-$t$ ancestor of every $w\in\mathcal{U}_{s,i}$ lies in $\mathcal{U}_{t,i}$. If $w$ is a depth-$s$ descendant of $u$, the definition of the induced
distribution gives $p_w/p_u=\Pr\left[\mathcal{A}\xrightarrow{s}w\,\middle|\,\mathcal{A}\xrightarrow{t}u\right]$. The events on the right are disjoint for distinct depth-$s$ descendants.
Consequently,
\[
\alpha_u=
\sum_{w\in\mathcal{W}_s(u)}
\Pr\left[
\mathcal{A}\xrightarrow{s}w
\,\middle|\,
\mathcal{A}\xrightarrow{t}u
\right] =
\Pr_{W\xleftarrow{s}\mathcal{A}}
\left[i\in\hat{D}(W)
\,\middle|\,
\mathcal{A}\xrightarrow{t}u\right]
\leq
r^{s-t},
\]
where the final inequality is the geometric-indistinguishability property.

The vectors $\{\boldsymbol{\phi}_u:u\in\mathcal{U}_{t,i}\}$ are
orthonormal. Therefore, for every
$\boldsymbol{z}\in\mathbb{R}^{\mathcal{F}}$,
\begin{align*}
\norm{\mathbf{Q}_t\mathbf{Q}_s\mathbf{Q}_t
\boldsymbol{z}}_2^2
&=
\sum_{u\in\mathcal{U}_{t,i}}
\alpha_u^2
\left|\left\langle\boldsymbol{\phi}_u,
\boldsymbol{z}\right\rangle\right|^2 \\
&\leq
r^{2(s-t)}
\sum_{u\in\mathcal{U}_{t,i}}
\left|\left\langle\boldsymbol{\phi}_u,
\boldsymbol{z}\right\rangle\right|^2 \\
&\leq
r^{2(s-t)}\norm{\boldsymbol{z}}_2^2
\qquad\text{(By Bessel's inequality)}.
\end{align*}
It follows that
$\norm{\mathbf{Q}_t\mathbf{Q}_s\mathbf{Q}_t}\leq r^{s-t}$.

Finally, symmetry and idempotence of $\mathbf{Q}_s$ give
\[
\norm{\mathbf{Q}_s\mathbf{Q}_t}^2=
\norm{(\mathbf{Q}_s\mathbf{Q}_t)^{\mathsf{T}}
(\mathbf{Q}_s\mathbf{Q}_t)} =
\norm{\mathbf{Q}_t^{\mathsf{T}}
\mathbf{Q}_s^{\mathsf{T}}\mathbf{Q}_s\mathbf{Q}_t} =
\norm{\mathbf{Q}_t\mathbf{Q}_s\mathbf{Q}_t} \leq r^{s-t}.
\]
Taking square roots proves
$\norm{\mathbf{Q}_s\mathbf{Q}_t}\leq r^{(s-t)/2}$ when $s>t$.

If $s<t$, apply the case just proved with $s$ and $t$ interchanged. Since
the spectral norm is invariant under transposition and both matrices are
symmetric,
\[
\norm{\mathbf{Q}_s\mathbf{Q}_t}
=
\norm{(\mathbf{Q}_s\mathbf{Q}_t)^{\mathsf{T}}}
=
\norm{\mathbf{Q}_t\mathbf{Q}_s}
\leq
r^{(t-s)/2}.
\]
Combining the three cases proves the claim.
\end{proof}

\begin{lemma}[Upper bound of the denominator] \label{lem:denominator-ub}
    For the matrix $\boldsymbol{\Gamma}$ in Definition \ref{def:adversary-matrix}, we have
    $$
    \max_{i\in I}\norm{\boldsymbol{\Gamma}\circ \boldsymbol{\Delta}_i} \leq \frac{1 + \sqrt{r}}{1 - \sqrt{r}}.
    $$
\end{lemma}

\begin{proof}
Fix $i\in I$. By Lemmas~\ref{lem:idempotent} and
\ref{lem:geometric-supression}, the matrices $\mathbf{Q}_0,\ldots,
\mathbf{Q}_{L - 1}$ satisfy Lemma~\ref{lem:summation-spectral-upperbound} with
its decay parameter equal to $\sqrt{r}$. Hence
\[
\norm{\mathbf{Q}}
\leq
\frac{1+\sqrt{r}}{1-\sqrt{r}}.
\]
Together with
$\norm{\boldsymbol{\Gamma}\circ\boldsymbol{\Delta}_i}
\leq\norm{\mathbf{Q}}$, this proves the claimed bound for every $i$.
\end{proof}

\begin{proof}[Proof of Theorem \ref{thm:tree-adversary}]
    Any algorithm for $f$ also computes its restriction $f|_{\mathcal{F}}$, so
    $Q_\epsilon(f)\geq Q_\epsilon(f|_{\mathcal{F}})$. It follows immediately from Lemma \ref{lem:numerator-lb} and \ref{lem:denominator-ub} that
    $$
    Q_\epsilon(f|_{\mathcal{F}}) \geq \left(1 - 2\sqrt{\epsilon(1 - \epsilon)}\right)\frac{\norm{\boldsymbol{\Gamma}}}{\max_{i\in I}\norm{\boldsymbol{\Gamma}\circ \boldsymbol{\Delta}_i}} \geq \left(1 - 2\sqrt{\epsilon(1 - \epsilon)}\right)\frac{1 - \sqrt r}{1 + \sqrt r}\sum_{t = 0}^{L - 1}a_t.
    $$
\end{proof}

\subsection{Proof of the Main Theorem}

In this subsection, we will use the tree-filtration adversary method to prove Theorem \ref{thm:main-theorem}. By
Lemmas~\ref{lem:unique_fixed_point} and~\ref{lem:monotonicity}, every
$\boldsymbol{f}_\pi\in\mathcal{F}_{n,k}$ is monotone and has the unique fixed point
$\boldsymbol{x}_m^\pi$. We therefore consider the restricted search problem
\[
\operatorname{Fix}_{n,k}\colon\mathcal{F}_{n,k}\longrightarrow[n]^k,
\qquad
\operatorname{Fix}_{n,k}(\boldsymbol{f}_\pi):=\boldsymbol{x}_m^\pi.
\]
A lower bound for this restricted problem is automatically a lower bound for \textsc{Tarski}$(n,k)$. 

For completeness, we explain how the quantum query model from Section \ref{subsec:computation_model} applies to the computation of $\operatorname{Fix}_{n, k}$. Let $\Sigma = \Pi = I = [n]^k$, and $\oplus$ denote the coordinate-wise addition modulo $n$ on $[n]^k$, making $(\Sigma, \oplus)$ an abelian group. The domain $S = \mathcal{F}_{n, k}$ is a subset of $\Sigma^I$. The input instance $\boldsymbol{f}_\pi\in S$ is not given explicitly but can be queried by $\mathbf{O}_{\boldsymbol{f}_\pi}$, where $\mathbf{O}_{\boldsymbol{f}_\pi}\ket{\boldsymbol{x}, \boldsymbol{z}} = \ket{\boldsymbol{x}, \boldsymbol{z}\oplus \boldsymbol{f}_\pi(\boldsymbol{x})}$ for any $\boldsymbol{x}\in I = [n]^k$ and $\boldsymbol{z}\in \Sigma = [n]^k$. The goal is to compute $\operatorname{Fix}_{n, k}(\boldsymbol{f}_\pi)$, namely the unique fixed point of $\boldsymbol{f}_\pi$.

Throughout the construction below, we retain the normalized parameters from
Section~\ref{sec:hard_instance}: $N=|\textsf{T}|=2m$, $k$ is even, and
$n-1=B=\sum_{r=0}^{U}4^r$. For a sequence $\boldsymbol{p}$ and a token $a$, write
$\boldsymbol{p}\mathbin{\|}a$ for
$\boldsymbol{p}\mathbin{\|}(a)$. Also write
$\boldsymbol{q}^{\mathrm{rev}}:=(q_{|\boldsymbol{q}|-1},\ldots,q_0)$.

First, we construct a filtration tree for $\mathcal{F}_{n, k}$ using a hierarchical prefix-suffix restriction structure. For every $0\leq t\leq m$ and every
$(\boldsymbol{p},\boldsymbol{q})\in
\textsf{SPair}_{\textsf{T}}^t$, let $\mathcal{B}(\boldsymbol{p},\boldsymbol{q})
:=\textsf{BP}_{\textsf{T}}\cap\left(\left.\textsf{P}_{\textsf{T}}\right|_{\boldsymbol{p},\boldsymbol{q}}
\right)$.

\begin{definition}[Filtration tree for $\mathcal{F}_{n,k}$]
\label{def:tarski-filtration-tree}
The rooted tree
$\mathcal{T}_{n,k}=(\mathcal{V}_{n,k},\mathcal{E}_{n,k},r_{n,k})$ is
specified as follows.
\begin{itemize}
    \item \textbf{The vertex set.} The vertices are the families of
    balanced completions determined by balanced prefix--suffix pairs:
    \[
    \mathcal{V}_{n,k}
    :=
    \bigcup_{t=0}^{m}
    \left\{
    \mathcal{B}(\boldsymbol{p},\boldsymbol{q}):
    (\boldsymbol{p},\boldsymbol{q})
    \in\textsf{SPair}_{\textsf{T}}^t
    \right\}.
    \]
    The vertex $\mathcal{B}(\boldsymbol{p},\boldsymbol{q})$ associated with
    a pair in $\textsf{SPair}_{\textsf{T}}^t$ has depth $t$. For brevity,
    we denote this vertex by $u_{\boldsymbol{p},\boldsymbol{q}}$.

    \item \textbf{The edge set.} Define
    \[
    \mathcal{E}_{n,k}
    :=
    \bigcup_{t=0}^{m-1}
    \left\{
    \left(
    u_{\boldsymbol{p},\boldsymbol{q}},
    u_{\boldsymbol{p}\mathbin{\|}a,\,
       \boldsymbol{q}\mathbin{\|}b}
    \right):
    (\boldsymbol{p},\boldsymbol{q})
    \in\textsf{SPair}_{\textsf{T}}^t,\quad
    a,b\in\textsf{T},\quad
    (\boldsymbol{p}\mathbin{\|}a,
     \boldsymbol{q}\mathbin{\|}b)
    \in\textsf{SPair}_{\textsf{T}}^{t+1}
    \right\}.
    \]

    \item \textbf{The root.} Define
    $r_{n,k}:=u_{\varepsilon,\varepsilon}$, corresponding to
    $\mathcal{B}(\varepsilon,\varepsilon)=\textsf{BP}_{\textsf{T}}$.
\end{itemize}
\end{definition}

Figure~\ref{fig:decision-tree} visualizes the tree  $\mathcal{T}_{2, 4}$.

\begin{figure}[t]
    \centering
    \resizebox{\textwidth}{!}{%
\begin{tikzpicture}[
    tree node/.style={
        circle,
        fill=black,
        inner sep=1.35pt
    },
    node label/.style={
        font=\scriptsize,
        text=black
    },
    leaf label/.style={
        font=\scriptsize,
        text=black,
        anchor=north
    },
    pi label/.style={
        font=\scriptsize,
        text=blue!70!black
    },
    tree edge/.style={
        draw=black!68,
        line width=0.38pt,
        shorten <=1.3pt,
        shorten >=1.3pt
    },
    pi arrow/.style={
        draw=blue!70!black,
        line width=0.65pt,
        -{Stealth[length=3.2pt,width=3.2pt]}
    },
    level line/.style={
        draw=yellow!42!orange,
        dashed,
        dash pattern=on 3pt off 2.5pt,
        line width=0.45pt
    },
    level label/.style={
        font=\small,
        anchor=east
    }
]

% Horizontal depth lines
\draw[level line] (-0.75,6.4) -- (23.75,6.4);
\draw[level line] (-0.75,4.2) -- (23.75,4.2);
\draw[level line] (-0.75,2.1) -- (23.75,2.1);

\node[level label] at (-1.00,6.4) {$\mathrm{dep}=0$};
\node[level label] at (-1.00,4.2) {$\mathrm{dep}=1$};
\node[level label] at (-1.00,2.1) {$\mathrm{dep}=2$};
\node[level label, text=blue!70!black]
    at (-1.00,0.30) {$\pi=$};

% Root
\node[tree node] (root) at (11.5,6.4) {};
\node[node label, anchor=south]
    at ([yshift=3pt]root.north)
    {$\varepsilon,\varepsilon$};

% #1 = node name
% #2 = x-coordinate
% #3 = displayed label
\newcommand{\parentnode}[3]{%
    \coordinate (#1-pos) at (#2,4.2);
    \draw[tree edge] (root) -- (#1-pos);
    \node[tree node] (#1) at (#1-pos) {};
    \node[node label, anchor=west]
        at ([xshift=3pt]#1.east) {$#3$};
}

% #1 = leaf name
% #2 = x-coordinate
% #3 = leaf label ab,cd
% #4 = permutation abdc
% #5 = parent node
\newcommand{\leafnode}[5]{%
    \coordinate (#1-pos) at (#2,2.1);
    \draw[tree edge] (#5) -- (#1-pos);
    \node[tree node] (#1) at (#1-pos) {};

    \node[leaf label]
        at ([yshift=-3pt]#1.south) {$#3$};

    \draw[pi arrow]
        (#2,1.48) -- (#2,0.68);

    \node[pi label] at (#2,0.30) {$#4$};
}

% Depth 1
\parentnode{p01}{0.5}{0,1}
\parentnode{p02}{2.5}{0,2}
\parentnode{p03}{4.5}{0,3}

\parentnode{p10}{6.5}{1,0}
\parentnode{p12}{8.5}{1,2}
\parentnode{p13}{10.5}{1,3}

\parentnode{p20}{12.5}{2,0}
\parentnode{p21}{14.5}{2,1}
\parentnode{p23}{16.5}{2,3}

\parentnode{p30}{18.5}{3,0}
\parentnode{p31}{20.5}{3,1}
\parentnode{p32}{22.5}{3,2}

% Depth 2 and the permutation row

% Children of 0,1
\leafnode{l0213}{0}{02,13}{0231}{p01}
\leafnode{l0312}{1}{03,12}{0321}{p01}

% Children of 0,2
\leafnode{l0123}{2}{01,23}{0132}{p02}
\leafnode{l0321}{3}{03,21}{0312}{p02}

% Children of 0,3
\leafnode{l0132}{4}{01,32}{0123}{p03}
\leafnode{l0231}{5}{02,31}{0213}{p03}

% Children of 1,0
\leafnode{l1203}{6}{12,03}{1230}{p10}
\leafnode{l1302}{7}{13,02}{1320}{p10}

% Children of 1,2
\leafnode{l1023}{8}{10,23}{1032}{p12}
\leafnode{l1320}{9}{13,20}{1302}{p12}

% Children of 1,3
\leafnode{l1032}{10}{10,32}{1023}{p13}
\leafnode{l1230}{11}{12,30}{1203}{p13}

% Children of 2,0
\leafnode{l2103}{12}{21,03}{2130}{p20}
\leafnode{l2301}{13}{23,01}{2310}{p20}

% Children of 2,1
\leafnode{l2013}{14}{20,13}{2031}{p21}
\leafnode{l2310}{15}{23,10}{2301}{p21}

% Children of 2,3
\leafnode{l2031}{16}{20,31}{2013}{p23}
\leafnode{l2130}{17}{21,30}{2103}{p23}

% Children of 3,0
\leafnode{l3102}{18}{31,02}{3120}{p30}
\leafnode{l3201}{19}{32,01}{3210}{p30}

% Children of 3,1
\leafnode{l3012}{20}{30,12}{3021}{p31}
\leafnode{l3210}{21}{32,10}{3201}{p31}

% Children of 3,2
\leafnode{l3021}{22}{30,21}{3012}{p32}
\leafnode{l3120}{23}{31,20}{3102}{p32}

\end{tikzpicture}%
}
    \caption{Visualization of the filtration tree $\mathcal{T}_{2, 4}$. For simplicity, $u_{\boldsymbol{p},\boldsymbol{q}}$ is written as $\boldsymbol{p}, \boldsymbol{q}$.}
    \label{fig:decision-tree}
\end{figure}
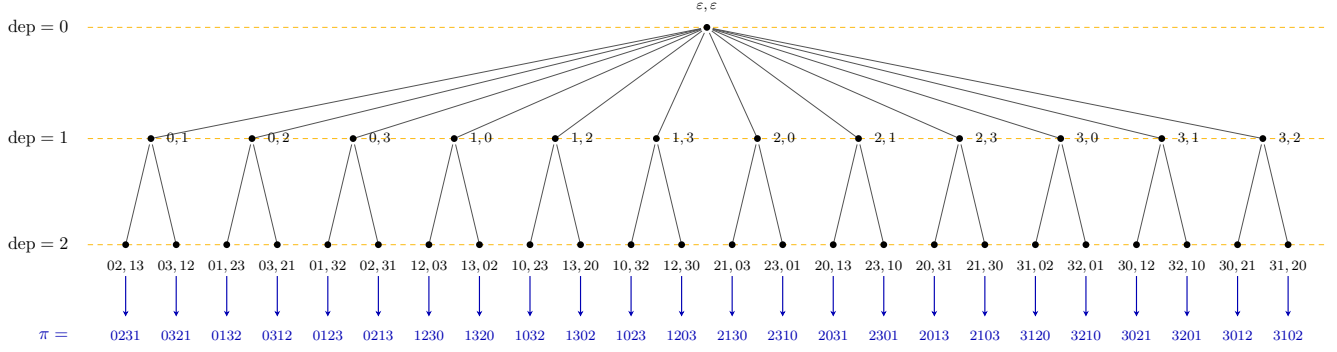 

\begin{lemma}
\label{lem:tarski-filtration-tree}
$\mathcal{T}_{n,k}$ is an $m$-equal-depth hard-instance filtration
tree for $\mathcal{F}_{n,k}$. Moreover,
\[
\mathcal{F}(u_{\boldsymbol{p},\boldsymbol{q}})
=
\{\boldsymbol{f}_\pi:\pi\in\mathcal{B}(\boldsymbol{p},\boldsymbol{q})\}.
\]
\end{lemma}

\begin{proof}
Every nonroot node has a unique parent, obtained by deleting the last entry
of each of its two sequences. This parent obviously remains a balanced
prefix--suffix pair.

Consider a node $u_{\boldsymbol{p},\boldsymbol{q}}$ at depth $t<m$. For
each layer $r$, exactly
$k-2\boldsymbol{d}(\boldsymbol{p})_r$ tokens from that layer remain unused. These
remaining counts are even, and their sum is $N-2t>0$. Thus some layer has
at least two unused tokens, which can be appended to
$\boldsymbol{p}$ and $\boldsymbol{q}$ to form a child. It follows that the
leaves are precisely the nodes at depth $m$.

If $(\boldsymbol{p},\boldsymbol{q})\in
\textsf{SPair}_{\textsf{T}}^m$, then
$\boldsymbol{p}\mathbin{\|}\boldsymbol{q}^{\mathrm{rev}}$ uses every token
exactly once and is balanced. Conversely, every
$\pi\in\textsf{BP}_{\textsf{T}}$ determines the leaf
$u_{\boldsymbol{p}_m(\pi),\boldsymbol{q}_m(\pi)}$. The injectivity of
$\pi\mapsto \boldsymbol{f}_\pi$, established in Section~\ref{sec:hard_instance}, shows
that $\mathrm{ins}$ is a bijection from the leaves onto
$\mathcal{F}_{n,k}$. The same correspondence below an arbitrary node gives
the relation we want.
\end{proof}

We have the following simplified difference criterion on $\mathcal{T}_{n, k}$:

\begin{definition}[Simplified difference criterion on $\mathcal{T}_{n, k}$] \label{def:difference-criterion}
    Let $\ell \leq m$, $(\boldsymbol{p}, \boldsymbol{q})\in \textsf{SPair}_{\textsf{T}}^{\ell}$. Define
    $$
    \hat{D}(u_{\boldsymbol{p}, \boldsymbol{q}}) = \{\boldsymbol{x} : \norm{\boldsymbol{x}} \leq W / 2 \wedge \boldsymbol{c}(\boldsymbol{p}) \preceq \boldsymbol{x}\} \cup \{\boldsymbol{x} : \norm{\boldsymbol{x}} > W / 2 \wedge \boldsymbol{x} \preceq B\boldsymbol{1} - \boldsymbol{c}(\boldsymbol{q})\}.
    $$
\end{definition}

\begin{lemma} \label{lem:difference-criterion}
    The map $\hat{D}$ in Definition~\ref{def:difference-criterion} is a simplified difference criterion.
\end{lemma}
\begin{proof}
    Fix a node $u_{\boldsymbol{p},\boldsymbol{q}}$ at depth $\ell$ and two
    instances $\boldsymbol{f}_\pi,\boldsymbol{f}_\sigma\in
    \mathcal{F}(u_{\boldsymbol{p},\boldsymbol{q}})$. We first prove $D(\boldsymbol{f}_\pi,\boldsymbol{f}_\sigma)
    \subseteq\hat{D}(u_{\boldsymbol{p},\boldsymbol{q}})$. It is enough to show that
    $\boldsymbol{x}\notin
    \hat{D}(u_{\boldsymbol{p},\boldsymbol{q}})$ implies
    $\boldsymbol{f}_\pi(\boldsymbol{x})=\boldsymbol{f}_\sigma(\boldsymbol{x})$.
    
    Suppose first that $\norm{\boldsymbol{x}}\leq W/2$. Then
    $\boldsymbol{c}(\boldsymbol{p})\npreceq\boldsymbol{x}$. The two spines
    share the vertices with indices $0,\ldots,\ell$, and $\boldsymbol{x}_\ell^\pi=\boldsymbol{x}_\ell^\sigma
    =\boldsymbol{c}(\boldsymbol{p})$. Since each spine is increasing, no vertex with index at least $\ell$ lies
    below $\boldsymbol{x}$. Hence
    $a:=\rho_\pi(\boldsymbol{x})=\rho_\sigma(\boldsymbol{x})<\ell$.
    Because $\ell\leq m$, we have $a<m$, so $h(a)=a+1\leq\ell$. Therefore,
    $\boldsymbol{f}_\pi(\boldsymbol{x})=\boldsymbol{x}_{a+1}^\pi
    =\boldsymbol{x}_{a+1}^\sigma=\boldsymbol{f}_\sigma(\boldsymbol{x})$.
    
    Suppose instead that $\norm{\boldsymbol{x}}>W/2$.  Then $\boldsymbol{x} \not\preceq B\boldsymbol{1} - \boldsymbol{c}(\boldsymbol{q})$. The two spines share the vertices with indices $N - \ell, \ldots, N$ and $\boldsymbol{x}^\pi_{N - \ell} = \boldsymbol{x}_{N - \ell}^\sigma = B\boldsymbol{1} - \boldsymbol{c}(\boldsymbol{q})$. Since each spine is increasing, no vertex with index at most $N - \ell$ lies above $\boldsymbol{x}$. Hence $a := \rho_\pi(\boldsymbol{x}) = \rho_{\sigma}(\boldsymbol{x}) > N - \ell$. Because $N - \ell \geq N - m = m$, we have $a > m$, so $h(a) = a - 1 \geq N - \ell$. Therefore, $\boldsymbol{f}_\pi(\boldsymbol{x}) = \boldsymbol{x}^\pi_{a - 1} = \boldsymbol{x}^\sigma_{a - 1} = \boldsymbol{f}_\sigma(\boldsymbol{x})$.
    
    For clarity, the intuition behind the proof on the lower half of the lattice is visualized in Figure \ref{fig:difference-criterion}.
    \begin{figure}[H]
        \centering
        \resizebox{0.5\linewidth}{!}{%
\begin{tikzpicture}[
    >=Stealth,
    point/.style={
        circle,
        fill=black,
        draw=none,
        inner sep=0pt,
        minimum size=3.2pt
    },
    red point/.style={
        point,
        fill=red!80!black
    },
    axis/.style={
        draw=black,
        line width=0.55pt,
        -{Stealth[length=4.5pt,width=4pt]}
    },
    ordinary arrow/.style={
        draw=black,
        line width=0.65pt,
        -{Stealth[length=4pt,width=3.6pt]}
    },
    red arrow/.style={
        draw=red!80!black,
        line width=0.8pt,
        -{Stealth[length=4.2pt,width=3.8pt]}
    },
    continuation arrow/.style={
        draw=black,
        line width=0.65pt,
        dashed,
        dash pattern=on 3pt off 2.2pt,
        -{Stealth[length=4pt,width=3.6pt]}
    },
    every node/.style={
        font=\small
    }
]

% Important coordinates
\coordinate (O) at (0,0);
\coordinate (X) at (4.25,3.35);
\coordinate (P) at (4.25,1.85);
\coordinate (F) at (5.45,2.62);
\coordinate (C) at (6.95,3.82);
\coordinate (PiEnd) at (8.65,5.25);
\coordinate (SigmaEnd) at (8.85,4.48);

% Light-blue rectangle
\fill[blue!9] (O) rectangle (X);

% Coordinate axes, without labels
\draw[axis] (O) -- (9.45,0);
\draw[axis] (O) -- (0,5.95);

% Thick dashed upper and right boundaries
\draw[
    draw=blue!65!black,
    line width=3.2pt,
    dashed,
    dash pattern=on 7pt off 5pt
]
    (0,3.35) -- (X) -- (4.25,0);

% Label in the upper-left corner of the rectangle
\node[
    anchor=north west,
    text=blue!65!black
] at (0.16,3.17) {$\preceq x$};

% Monotone curved arrow from the origin to Pi(x)
% Both control points increase in both coordinates.
\draw[ordinary arrow]
    (0.08,0.08)
    .. controls (1.05,0.28) and (2.85,1.28) ..
    (P);

% Straight red arrow from x to Pi(x)
% The white underlay keeps it visible over the blue dashed boundary.
\draw[
    red arrow,
    preaction={
        draw=white,
        line width=2.0pt,
        shorten <=2.3pt,
        shorten >=2.3pt
    },
    shorten <=2.3pt,
    shorten >=2.3pt
]
    (X) -- (P);

% Short red arrow from Pi(x) to f_pi(x)=f_sigma(x)
\draw[
    red arrow,
    shorten <=2.2pt,
    shorten >=2.2pt
]
    (P) -- (F);

% Curved black arrow from f_pi(x)=f_sigma(x) to C(p)
\draw[
    ordinary arrow,
    shorten <=2.2pt,
    shorten >=2.2pt
]
    (F)
    .. controls (5.86,2.86) and (6.38,3.58) ..
    (C);

% Two different upward-right dashed continuations
\draw[
    continuation arrow,
    shorten <=2.2pt
]
    (C)
    .. controls (7.36,4.08) and (8.03,4.88) ..
    (PiEnd);

\draw[
    continuation arrow,
    shorten <=2.2pt
]
    (C)
    .. controls (7.45,3.94) and (8.20,4.12) ..
    (SigmaEnd);

% Solid points
\node[point]     at (X) {};
\node[red point] at (P) {};
\node[red point] at (F) {};
\node[point]     at (C) {};

% Point labels
\node[anchor=south west] at ([xshift=2pt,yshift=2pt]X)
    {$x$};

\node[
    anchor=north west,
    text=red!80!black
] at ([xshift=3pt,yshift=-2pt]P)
    {$\Pi(x)$};

\node[
    anchor=south west,
    text=red!80!black
] at ([xshift=3pt,yshift=2pt]F)
    {$f_\pi(x)=f_\sigma(x)$};

\node[
    anchor=south east
] at ([xshift=-3pt,yshift=3pt]C)
    {$\boldsymbol{C}(\boldsymbol{p})$};

% Labels at the ends of the dashed arrows
\node[anchor=south west]
    at ([xshift=2pt,yshift=1pt]PiEnd)
    {$\pi$};

\node[anchor=west]
    at ([xshift=3pt]SigmaEnd)
    {$\sigma$};

\end{tikzpicture}%
}
        \caption{{Visualization of the proof of Lemma \ref{lem:difference-criterion} on the lower half of the lattice.} This figure shows that $\norm{\boldsymbol{x}} \leq W / 2$ and $\boldsymbol{c}(\boldsymbol{p})\not\preceq \boldsymbol{x}$ implies $\boldsymbol{f}_\pi(\boldsymbol{x}) = \boldsymbol{f}_\sigma(\boldsymbol{x})$.}
        \label{fig:difference-criterion}
    \end{figure}
    
    It remains to verify $\hat{D}$ is nested along $\mathcal{T}$. If
    $u_{\boldsymbol{p}',\boldsymbol{q}'}$ is a descendant of
    $u_{\boldsymbol{p},\boldsymbol{q}}$, then
    $\boldsymbol{p}'$ and $\boldsymbol{q}'$ extend
    $\boldsymbol{p}$ and $\boldsymbol{q}$, respectively. Hence
    $\boldsymbol{c}(\boldsymbol{p})\preceq
    \boldsymbol{c}(\boldsymbol{p}')$ and
    $B\boldsymbol{1}-\boldsymbol{c}(\boldsymbol{q}')\preceq
    B\boldsymbol{1}-\boldsymbol{c}(\boldsymbol{q})$,
    which directly gives
    $\hat{D}(u_{\boldsymbol{p}',\boldsymbol{q}'})
    \subseteq
    \hat{D}(u_{\boldsymbol{p},\boldsymbol{q}})$.
\end{proof}

We now define the distinguishing random walk $\mathcal{A}$. For every node $u_{\boldsymbol{p}, \boldsymbol{q}}$ and $r\in [U + 1]$, let
$$
R_r(\boldsymbol{p}, \boldsymbol{q}) := \{a\in \textsf{T} : r_a = r,~a\notin \boldsymbol{p}, a\notin \boldsymbol{q}\}
$$
denote the set of remaining tokens at each layer, where membership in a sequence means membership among its entries, and
$$
M_r := |R_r(\boldsymbol{p}, \boldsymbol{q})| = k - 2\boldsymbol{d}(\boldsymbol{p})_r.
$$

At an internal node, set
$$
Z := \sum_{r = 0}^U M_r4^r.
$$
From the assumption that $k$ is even, all $M_r$ are even, and at least one is positive; hence $Z > 0$.

\begin{definition}[Distinguishing random walk on $\mathcal{T}_{n, k}$] \label{def:distinguishing-walk}
At a non-leaf node $u_{\boldsymbol{p}, \boldsymbol{q}}$, the random walk $\mathcal{A}$ first chooses a layer $r$ with $M_r > 0$, assigning it probability $M_r4^r / Z$. It then chooses ${p}^+$ and ${q}^+$ uniformly from $R_r(\boldsymbol{p}, \boldsymbol{q})$ without replacing. The next node is
$
u_{\boldsymbol{p} \mathbin{\|} p^+, \boldsymbol{q} \mathbin{\|} q^+}.
$
\end{definition}

For distinct $p^+, q^+\in R_r(\boldsymbol{p}, \boldsymbol{q})$, the corresponding transition probability is
\begin{align}
P_{\mathcal{A}}\left(u_{\boldsymbol{p} \mathbin{\|} p^+, \boldsymbol{q} \mathbin{\|} q^+} \mid u_{\boldsymbol{p}, \boldsymbol{q}}\right) = \frac{4^r}{Z(M_r - 1)}. \label{eq:transition_probability}
\end{align}
Thus, every child has positive probability and the outgoing probabilities sum to one. Moreover, every unused token $a = (j, r)$ satisfies
\begin{align}
    \Pr_{\mathcal{A}}[p^+ = a\mid u_{\boldsymbol{p}, \boldsymbol{q}}] = \Pr_{\mathcal{A}}[q^+ = a \mid u_{\boldsymbol{p}, \boldsymbol{q}}] = \frac{4^r}{Z}. \label{eq:token-marginal}
\end{align}

For a node $u_{\boldsymbol{p}, \boldsymbol{q}}$, denote its set of possible fixed points by $\mathcal{Z}(\boldsymbol{p}, \boldsymbol{q}) := \{\boldsymbol{x}_m^\pi : \pi\in \mathcal{B}(\boldsymbol{p}, \boldsymbol{q})\}$. Let $H(\boldsymbol{p}, \boldsymbol{q}) := |\mathcal{Z}(\boldsymbol{p}, \boldsymbol{q})|$.

\begin{lemma} \label{lem:A-uniform-fixed-point}
    Conditioned on $\mathcal{A}$ reaching $u_{\boldsymbol{p}, \boldsymbol{q}}$, the fixed point associated with the random leaf is uniform on $\mathcal{Z}(\boldsymbol{p}, \boldsymbol{q})$. Furthermore,
    $$
    H(\boldsymbol{p}, \boldsymbol{q}) = \prod_{r = 0}^U\binom{M_r}{M_r / 2}.
    $$
\end{lemma}

\begin{proof}
    A balanced leaf below $u_{\boldsymbol{p},\boldsymbol{q}}$ places exactly $M_r/2$ tokens from $R_r(\boldsymbol{p},\boldsymbol{q})$ in the remaining part of the prefix. If $A_r \subseteq R_r(\boldsymbol{p}, \boldsymbol{q})$ ($|A_r| = M_r / 2$) denotes this subset, then its midpoint is
    $$
    \mathrm{Mid}(A_0, \ldots, A_U) := \boldsymbol{c}(\boldsymbol{p}) + \sum_{r = 0}^U\sum_{a\in A_r}\boldsymbol{v}_a.
    $$

    Conversely, every tuple $(A_0, \ldots, A_U)$ satisfying these conditions can be completed to a balanced leaf by pairing the tokens in $A_r$ with those in $R_r\setminus A_r$ within each layer. Moreover, for each coordinate $j$, the base-$4$ expansion of $(\boldsymbol{x}^*-\boldsymbol{c}(\boldsymbol{p}))_j$ uniquely determines whether $(j,r)\in A_r$ for every $r$. Hence this correspondence is a bijection between the tuples $(A_0,\ldots,A_U)$ and the fixed points hidden in the subtree, which proves the formula for $H(\boldsymbol{p},\boldsymbol{q})$.

    It remains to prove the uniformity. Order the tokens in each $R_r(\boldsymbol{p}, \boldsymbol{q})$ by their directions, that is $(j_1, r) < (j_2, r)$ if $j_1 < j_2$, and let $A_r^{\mathrm{ref}}$ consist of the first $M_r / 2$ tokens in $R_r(\boldsymbol{p}, \boldsymbol{q})$ in this order. We compare an arbitrary tuple $(A_0, \ldots, A_U)$ with this representative tuple $(A_0^{\mathrm{ref}}, \ldots, A_U^{\mathrm{ref}})$.

    For each $r$, define a bijection $\varphi_r\colon \ R_r(\boldsymbol{p}, \boldsymbol{q})\rightarrow R_r(\boldsymbol{p}, \boldsymbol{q})$ by matching, in the order defined above, the $i$-th token of $A_r$ with the $i$-th token of $A_r^{\mathrm{ref}}$, and the $i$-th token of $R_r(\boldsymbol{p}, \boldsymbol{q}) \setminus A_r$ with the $i$-th token of $R_r(\boldsymbol{p}, \boldsymbol{q})\setminus A_r^{\mathrm{ref}}$. Applying these bijections token-wise to every path generated by the random walk below $u_{\boldsymbol{p}, \boldsymbol{q}}$ gives a bijection between the paths producing $\mathrm{Mid}(A_0, \ldots, A_U)$ and those producing $\mathrm{Mid}(A_0^{\mathrm{ref}}, \ldots, A_U^{\mathrm{ref}})$.

    This bijection couples the layer $r$ chosen at every step. Consequently, corresponding states have the same number $M_r$ of remaining tokens in every layer and hence the same value of $Z$. By Equation \eqref{eq:transition_probability}, the transition probability only depends on $r, M_r$ and $Z$. Hence, the transitions have the same probability. Thus corresponding paths have equal probabilities. It follows that every tuple $(A_0, \ldots, A_U)$, and therefore every $\operatorname{Mid}(A_0, \ldots, A_U)$ occurs with the same probability.
\end{proof}

By Lemma \ref{lem:A-uniform-fixed-point}, establishing $1/2$-anti-concentration now reduces to the simple fact that every non-leaf subtree contains at least two distinct possible points.

\begin{lemma}[Anti-concentration] \label{lem:A-anti-concentration}
    Let $a_0 = a_1 = \cdots  = a_{m - 1} = 1/2$. Then $\mathcal{A}$ satisfies the
    $\boldsymbol{a}=(a_0,\ldots,a_{m - 1})$-anti-concentration property for
    $\operatorname{Fix}_{n,k}$.
\end{lemma}

\begin{proof}
    At a node of depth $t < m$,
    $$
    \sum_{r = 0}^U M_r = N - 2t = 2(m - t) > 0.
    $$
     Since every $M_r$ is even, some layer satisfies $M_r \geq 2$. Therefore, $H(\boldsymbol{p}, \boldsymbol{q}) \geq 2$. By Lemma \ref{lem:A-uniform-fixed-point}, for every fixed descendant leaf $\ell$ and a random descendant leaf $V\xleftarrow{m}\mathcal{A}$, 
     $$
     \Pr_{V\xleftarrow{m}\mathcal{A}}\left[\mathrm{Fix}_{n, k}(\mathrm{ins}(V)) = \mathrm{Fix}_{n, k}(\mathrm{ins}(\ell)) ~\middle|~ \mathcal{A}\xrightarrow{t} u_{\boldsymbol{p}, \boldsymbol{q}}\right] = \frac{1}{H(\boldsymbol{p}, \boldsymbol{q})} \leq 1/2.
     $$

     Taking the complementary event gives the required lower bound $a_t = 1/2$. 
\end{proof}

The next lemma gives the one-step contraction used to establish geometric-indistinguishability.

\begin{lemma}[$2/3$-Compression for Adjacent Layers] \label{lem:2/3-compression}
    Let $u=u_{\boldsymbol{p},\boldsymbol{q}}$ have depth $t<m$, and let $V_{t+1}\xleftarrow{t+1}\mathcal{A}$ be its random child. For every $\boldsymbol{x}\in[n]^k$,
    $$
    \Pr_{V_{t+1}\xleftarrow{t+1}\mathcal{A}}\left[\boldsymbol{x}\in \hat{D}(V_{t + 1})~\middle|~\mathcal{A}\xrightarrow{t}u\right]\leq \frac 23.
    $$
\end{lemma}

\begin{proof}
    If $\boldsymbol{x}\notin \hat{D}(u)$, the probability is zero by the nesting of $\hat{D}$. Assume henceforth that $\boldsymbol{x}\in \hat{D}(u)$. First observe that
    \begin{align}
        Z = \sum_{r = 0}^U 
\bigl(k-2\boldsymbol{d}(\boldsymbol{p})_r\bigr)4^r = W - 2\norm{\boldsymbol{c}(\boldsymbol{p})}. \label{eq:normalization-Z}
    \end{align}
    The same equation applies for $\boldsymbol{q}$ since $(\boldsymbol{p}, \boldsymbol{q})$ is balanced. We also use the elementary bound
    \begin{align}
        \sum_{\substack{r\in[U+1]:\,4^r\leq A}}4^r \leq \sum_{i = 0}^\infty \frac{1}{4^i}A =  \frac43 A \qquad(A\geq0). \label{eq:geometric-slack}
    \end{align}

    Suppose first that $\norm{\boldsymbol{x}} \leq W / 2$. Since $\boldsymbol{x}\in \hat{D}(u)$, the per-coordinate overheads $A_j := x_j - \boldsymbol{c}(\boldsymbol{p})_j$ are nonnegative. A child contains $\boldsymbol{x}$ in its simplified difference criterion when $p^+ = (j, r)$ satisfies $4^r \leq A_j$. By Equation \eqref{eq:token-marginal}, \begin{align*}
        \Pr_{V_{t+1}\xleftarrow{t+1}\mathcal{A}}\left[\boldsymbol{x}\in \hat{D}(V_{t + 1}) \mid \mathcal{A}\xrightarrow{t}u\right] &= \frac{1}{Z}\sum_{r = 0}^U\sum_{(j, r)\in R_r(\boldsymbol{p}, \boldsymbol{q})} 4^r\mathds{I}\{4^r \leq A_j\} \\
        &\leq \frac 1Z \sum_{j\in [k]}\sum_{\substack{r\in[U+1]:\,4^r\leq A_j}}4^r \\
        &\leq \frac{4}{3Z}\sum_{j\in [k]}A_j \\
        &= \frac{4}{3Z}\left(\norm{\boldsymbol{x}} - \norm{\boldsymbol{c}(\boldsymbol{p})}\right) \\
        &\leq \frac{4}{3Z}\left(\frac W2 - \norm{\boldsymbol{c}(\boldsymbol{p})}\right) = \frac 23,
    \end{align*}
    where the last equality uses Equation \eqref{eq:normalization-Z}.

    Suppose instead $\norm{\boldsymbol{x}} > W / 2$. Since $\boldsymbol{x}\in \hat{D}(u)$, the per-coordinate overheads $A_j := (B - \boldsymbol{c}(\boldsymbol{q})_j) - x_j$ are nonnegative. A child contains $\boldsymbol{x}$ in its simplified difference criterion when ${q}^+ = (j, r)$ satisfies $4^r \leq A_j$. By Equation \eqref{eq:token-marginal},
    \begin{align*}
        \Pr_{V_{t+1}\xleftarrow{t+1}\mathcal{A}}\left[\boldsymbol{x}\in \hat{D}(V_{t + 1}) \mid \mathcal{A}\xrightarrow{t}u\right] &= \frac{1}{Z}\sum_{r = 0}^U\sum_{(j, r)\in R_r(\boldsymbol{p}, \boldsymbol{q})} 4^r\mathds{I}\{4^r \leq A_j\} \\
        &\leq \frac 1Z \sum_{j\in [k]}\sum_{\substack{r\in[U+1]:\,4^r\leq A_j}}4^r \\
        &\leq \frac{4}{3Z}\sum_{j\in [k]}A_j \\
        &= \frac{4}{3Z}\left((\norm{B\boldsymbol{1}} - \norm{\boldsymbol{c}(\boldsymbol{q})}) - \norm{\boldsymbol{x}}\right) \\
        &< \frac{4}{3Z}\left(\frac W2 - \norm{\boldsymbol{c}(\boldsymbol{q})}\right) = \frac 23,
    \end{align*}
    which completes the proof.
\end{proof}

\begin{lemma}[Geometric-indistinguishability] \label{lem:A-geometric-indistinguishability}
$\mathcal{A}$ satisfies the $2/3$-geometric-indistinguishability property.
\end{lemma}

\begin{proof}
    Fix depths $0\leq t < s \leq m$, a node $u$ at depth $t$, and a query point $\boldsymbol{x}\in [n]^k$. Let $V_d$ denote the random node at depth $d$. For $d\geq t$, let $E_d$ be the event $\boldsymbol{x}\in \hat{D}(V_d)$. For every $t \leq d < s$, conditioning on the depth-$d$ node and applying Lemma \ref{lem:2/3-compression} gives
    \begin{align*}
        \Pr_{V_{d+1}\xleftarrow{d+1}\mathcal{A}}[E_{d + 1} \mid V_t = u] &= \sum_{\substack{w\in\mathcal{V}_{n,k}\\\mathrm{dep}(w)=d}} \Pr[V_d=w\mid V_t=u]\, \Pr[E_{d+1}\mid V_d=w] \\
        &\leq \frac23 \sum_{\substack{w\in\mathcal{V}_{n,k} \\ \mathrm{dep}(w)=d}} \Pr[V_d=w\mid V_t=u]\, \mathds{I}\{\boldsymbol{x}\in\hat{D}(w)\} \\
        &=\frac23\Pr_{V_d\xleftarrow{d}\mathcal{A}}[E_d\mid V_t=u].
    \end{align*}
    For the first equality, recall that $\mathcal{A}$ is Markov, and hence $\Pr[E_{d + 1} \mid V_d = w, V_t = u] = \Pr[E_{d + 1} \mid V_d = w]$.

    Iterating this inequality yields
    $$
    \Pr_{V_{s}\xleftarrow{s}\mathcal{A}}[E_s\mid V_t = u] \leq \left(\frac 23\right)^{s - t} \Pr_{V_t\xleftarrow{t}\mathcal{A}}[E_t\mid V_t=u] \leq \left(\frac 23\right)^{s - t}.
    $$
    This is precisely the required geometric-indistinguishability bound.
\end{proof}

With all the results above, we can prove our main theorem.

\begin{proof}[Proof of Theorem \ref{thm:main-theorem}]
    The case where $n = 1$ is trivial. When $k = 1$, the problem reduces to ordered search \cite{chang2008complexity} and is $\Omega(\log n)$ by \cite{HPJ01BinarySearch}. Otherwise, for the normalized parameters, Lemma \ref{lem:tarski-filtration-tree} provides an $m$-equal-depth filtration tree for $\mathcal{F}_{n, k}$, Lemma \ref{lem:difference-criterion} provides a simplified difference criterion, and Lemmas \ref{lem:A-anti-concentration} and \ref{lem:A-geometric-indistinguishability} establish the two required properties of $\mathcal{A}$. We apply Theorem \ref{thm:tree-adversary} with
    $$
    L = m ,\qquad a_0 = \cdots = a_{m - 1} = 1/2, \qquad r=\frac 23.
    $$

    For every constant $0\leq \epsilon < 1/2$, this gives \begin{align*}
        Q_\epsilon(\mathrm{Fix}_{n, k}) &\geq \left(1 - 2\sqrt{\epsilon(1 - \epsilon)}\right)\frac{1 - \sqrt{2/3}}{1 + \sqrt{2/3}}\sum_{t = 0}^{m - 1} a_t \\
        &= \left(1 - 2\sqrt{\epsilon(1 - \epsilon)}\right)\frac{1 - \sqrt{2/3}}{1 + \sqrt{2/3}}\frac m2 \\
        &= \Omega(k\log n).
    \end{align*}

    It remains to restore the original parameters. For arbitrary $n, k\geq 2$, let $\bar{k}$ be the largest even integer at most $k$, let $\bar{U}$ be maximal subject to $1 + \sum_{r = 0}^{\bar{U}}4^r \leq n$, and set $\bar{n} = 1 + \sum_{r = 0}^{\bar{U}}4^r$. It follows that $4(\bar{n} - 1) = \sum_{r = 1}^{\bar{U} + 1}4^r > n - 2$ and $\bar{n} \leq n$, hence $\bar{n} = \Theta(n)$ and consequently $\bar{U} = \Theta(\log n)$. Moreover, $\bar{k} \geq k - 1$, hence $\bar{k} = \Theta(k)$. By monotonicity of the Tarski query complexity in both parameters and by the fact that every Tarski algorithm computes the restriction on the hard instance family,
    $$
    Q_\epsilon(\textsc{Tarski}(n, k)) \geq Q_\epsilon(\textsc{Tarski}(\bar{n}, \bar{k})) \geq Q_\epsilon(\mathrm{Fix}_{\bar{n}, \bar{k}}) = \Omega(\bar{k}\log \bar{n}) = \Omega(k\log n).
    $$

    This completes the proof.
\end{proof}

Finally, we observe that our lower bound is already tight for this family of hard instances. Indeed, under the normalized parameter setting, the iterative algorithm that starts from $\boldsymbol{0}$ and repeatedly applies $\boldsymbol{x}\leftarrow\boldsymbol{f}(\boldsymbol{x})$ reaches the fixed point within $m = k(U + 1) / 2 = O(k\log n)$ iterations. Therefore,
$$
Q_\epsilon(\operatorname{Fix}_{n, k}) = \Theta(k\log n).
$$

\section{Conclusion} \label{sec:conclusion}

We prove an $\Omega(k\log n)$ bounded-error quantum query lower bound for
\textsc{Tarski}$(n,k)$. To the best of our knowledge, this is the first quantum lower bound for general $n$ and $k$, capturing
their joint influence on the complexity. The bound is tight in two extremal
regimes: it matches the optimal $\Theta(k)$ complexity when $n=2$ and the
optimal $\Theta(\log n)$ complexity when $k=1$~\cite{BPRRandomizedLowerbound,chang2008complexity,HPJ01BinarySearch}. For $n< k$, our result improves the best previously known classical lower bounds. When $n\geq k$, a logarithmic-factor gap of $\log n / \log k$ remains.

Beyond this application, the Tree-Filtration Adversary Method we use in the proof may be of
independent interest. For a hierarchically organized family of hard instances
and a suitable distinguishing random walk, the method constructs a canonical
nonnegative adversary matrix and reduces its analysis to two probabilistic
properties: anti-concentration and geometric-indistinguishability. By
separating the combinatorial design of the hard instances from the spectral
analysis of the adversary matrix, this framework potentially provides a reusable route
to quantum query lower bounds for other problems with tree-structured instance
families.

\paragraph{Open Problems.} We leave the following open problems for future investigation:

\begin{enumerate}
    \item Can quantum algorithms achieve a speedup over classical algorithms for $\textsc{Tarski}(n,k)$? \cite{phillips2026quantum} has shown that even quantum algorithms cannot readily escape the nested binary-search structure to achieve a speedup over classical algorithms in the low-dimensional regime. Meanwhile, our $\Omega(k\log n)$ lower bound suggests that quadratic-speedup techniques such as the wildcard search in Appendix \ref{apdx:bpr25-upperbound} are unlikely to be effective. However, developing a coherent version of the \textit{decomposition lemma} \cite{dang2011computational,chen2022improved} may provide a path toward a quantum speedup over classical algorithms in high dimensions.
    
    \item What is the quantum query complexity of the multidimensional-herringbone instances \cite{branzei2025tarski}? This family of hard instances yields a $\Omega(k\log^2 n / \log k)$ classical lower bound, hence its quantum query complexity is worth studying.  A key step established in Lemma 18 of \cite{branzei2025tarski} implies potential quantum speedup: it shows that for any $0 \leq x \leq (n − 1)k$, one can find a spine vertex with $\ell_1$-norm $x$ using $O(k \log n)$ queries, and it is worth investigating whether an approach similar to the wildcard search in Appendix \ref{apdx:bpr25-upperbound} can yield a quadratic quantum speedup. This lemma is used to derive the $O(k\log n\log(nk))$ query classical upper bound in \cite{branzei2025tarski}, making it a natural starting point for exploring quantum improvements.

    \item Can we extend the Tree-Filtration Adversary Method? The current version of the method is a simple abstraction sufficient to handle the two applications in this paper. Possible extensions include, but are not limited to, allowing non-equal-depth filtration trees, permitting the map from leaves to hard instances to be surjective rather than bijective, allowing nonuniform anti-concentration, and introducing layer-specific or even subtree-specific geometric indistinguishability. 
\end{enumerate}

\paragraph{Acknowledgement.} We thank Yuhao Li for early disscusion in this project. Tongyang Li and Ziyi Yang are supported by the National Natural Science Foundation of China under Grant No.~62372006. We applied GPT-5.6 to polishing presentation and checking the clarity and completeness of proofs.

\newcommand{\etalchar}[1]{$^{#1}$}

\newpage
\appendix

\section{Quantum Query Complexity of a Previous Hard Instance Family} \label{apdx:bpr25-upperbound}

In this section, we prove the quantum query complexity of the hard instance family in \cite{BPRRandomizedLowerbound} is: for constant $0 < \epsilon < 1/2$,
$$
Q_\epsilon(\operatorname{Fix}^{\text{BPR}}_{n, k}) =\Omega\left(\frac{\sqrt{k}}{\log k}\right) ~ \text{and}~ Q_\epsilon(\operatorname{Fix}^{\text{BPR}}_{n, k}) = O(\sqrt{k}\log k\log n\log\log n).
$$

Assume $n, k\geq 2$. Recall that in \cite{BPRRandomizedLowerbound}, for any $\boldsymbol{a}\in [n]^k$, a corresponding hard instance $\boldsymbol{f}_{\boldsymbol{a}} \colon [n]^k\rightarrow [n]^k$ is defined coordinatewise by
$$
\boldsymbol{f}_{\boldsymbol{a}}(\boldsymbol{v})_i = \begin{cases}
    v_i - 1, & \text{if $v_i > a_i$ and $\forall j < i, v_j \leq a_j$,} \\
    v_i + 1, & \text{if $v_i < a_i$ and $\forall j < i, v_j \geq a_j$,} \\
    v_i,     & \text{otherwise.}
\end{cases}
$$

Denote the fixed point computation restricted to this family of hard instances by $\operatorname{Fix}^{\text{BPR}}_{n, k}$. We first prove that $\operatorname{Fix}^{\text{BPR}}_{n, k}$ admits an $\widetilde{O}(\sqrt k)$-query quantum algorithm.

We recover the hidden vector $\boldsymbol{a}$ by reducing each round of a coordinatewise binary search to the quantum search-with-wildcards problem of \cite{WildcardSearch}.

Fix a threshold vector $\boldsymbol{m} \in \{1, \ldots, n - 1\}^k$ and define $\boldsymbol{b}$ by
$$
b_i = \mathbf{1}\{a_i < m_i\}.
$$

We first show how to simulate wildcard queries to $\boldsymbol{b}$ using only a constant number of queries to $\mathbf{O}_{\boldsymbol{f}_{\boldsymbol{a}}}$.

For $S\subseteq [k]$, Algorithm~\ref{alg:judge-all-0} tests whether $b_i = 0$ for every $i \in S$ using one query.

\begin{algorithm}[H]
\caption{$\textsc{JudgeAllZero}(S, \boldsymbol{m})$}
\label{alg:judge-all-0}
\begin{algorithmic}[1]
    \Require{$S\subseteq [k], \boldsymbol{m}\in [n]^k$}
    \State Define $\boldsymbol{m}' \in [n]^k$ by
    $$
    m'_i = \begin{cases}
    m_i,& i\in S, \\
        0, & i\notin S.
    \end{cases}
    $$
    \State Query $\boldsymbol{v}\leftarrow \mathbf{O}_{\boldsymbol{f}_{\boldsymbol{a}}}(\boldsymbol{m}')$
    \If{$\exists i\in [k], v_i = m'_i - 1$}
        \State \textbf{return} False
    \Else
        \State \textbf{return} True
    \EndIf
\end{algorithmic}
\end{algorithm}

\begin{lemma} \label{lem:judge-all-0}
    Algorithm \ref{alg:judge-all-0} returns True if and only if $\forall i\in S, b_i = 0$.
\end{lemma}

\begin{proof}
    Note that the oracle decreases some coordinate only if $m'_i > a_i$ for some $i\in [k]$. Conversely, if at least one such coordinate exists, the first one necessarily satisfies the condition in the first case of the definition of $\boldsymbol{f}_{\boldsymbol{a}}$. Since $m'_i = 0\leq a_i$ outside $S$, Algorithm \ref{alg:judge-all-0} returns false if and only if $m_i > a_i$ for some $i\in S$, or equivalently, if and only if $b_i = 1$ for some $i\in S$.
\end{proof}

Similarly, Algorithm \ref{alg:judge-all-1} tests if $\forall i\in S, b_i = 1$ with $1$ query to $\mathbf{O}_{\boldsymbol{f}_{\boldsymbol{a}}}$.

\begin{algorithm}[H]
\caption{$\textsc{JudgeAllOne}(S, \boldsymbol{m})$}
\label{alg:judge-all-1}
\begin{algorithmic}[1]
    \Require{$S\subseteq [k], \boldsymbol{m}\in [n]^k$}
    \State Define $\boldsymbol{m}'$ by
    $$
    m'_i = \begin{cases}
        m_i - 1,& i\in S, \\
        n - 1, & \text{otherwise}.
    \end{cases}
    $$
    \State Query $\boldsymbol{v}\leftarrow \mathbf{O}_{\boldsymbol{f}_{\boldsymbol{a}}}(\boldsymbol{m}')$
    \If{$\exists i\in [k], v_i = m'_i + 1$}
        \State \textbf{return} False
    \Else
        \State \textbf{return} True
    \EndIf
\end{algorithmic}
\end{algorithm}

\begin{lemma} \label{lem:judge-all-1}
    Algorithm \ref{alg:judge-all-1} returns True if and only if $\forall i \in S, b_i = 1$.
\end{lemma}

\begin{proof}
    Note that the oracle increases some coordinate only if $m'_i < a_i$ for some $i\in [k]$. Conversely, if at least one such coordinate exists, the first one necessarily satisfies the condition in the second case of the definition of $\boldsymbol{f}_{\boldsymbol{a}}$. Since $m'_i = n - 1 \geq a_i$ outside $S$, Algorithm \ref{alg:judge-all-1} returns false if and only if $m_i \leq a_i$ for some $i\in S$, or equivalently, if and only if $b_i = 0$ for some $i\in S$.
\end{proof}

We can now implement the required wildcard oracle in Algorithm \ref{alg:wildcard-oracle}. 

\begin{algorithm}[H]
\caption{$\textsc{WildcardOracle}(S, \boldsymbol{x})$}
\label{alg:wildcard-oracle}
\begin{algorithmic}[1]
    \Require{$S\subseteq [k], \boldsymbol{x}\in [2]^k$}
    \State $S_1 \leftarrow \{i\in S : x_i = 1\}$ and $S_0 \leftarrow \{i\in S:x_i = 0\}$
    \State \textbf{return} $\textsc{JudgeAllOne}(S_1, \boldsymbol{m}) \wedge \textsc{JudgeAllZero}(S_0, \boldsymbol{m})$
\end{algorithmic}
\end{algorithm}

Algorithm \ref{alg:wildcard-oracle} returns True exactly when $\boldsymbol{x}_S = \boldsymbol{b}_S$. Each call uses two queries to $\mathbf{O}_{\boldsymbol{f}_{\boldsymbol{a}}}$. To run this algorithm coherently, it suffices to store its output and then run the algorithm once in reverse to uncompute the work registers. The number of queries to \(\mathbf{O}_{\boldsymbol{f}_{\boldsymbol{a}}}\) remains constant, and in particular $4$. By Theorem 1 of \cite{WildcardSearch}, the entire string $\boldsymbol{b}$ can therefore be recovered with $O(\sqrt{k}\log k)$ quantum queries to $\textsc{WildcardOracle}$ with constant probability.

Finally, maintain an interval $[l_i, r_i]$ of possible values for each coordinate $a_i$. In every round, choose the midpoint $m_i = \lceil(l_i + r_i) / 2\rceil$ for each nonsingleton interval $[l_i, r_i]$ ($m_i$ for $l_i = r_i$ can be chosen arbitrarily, we choose $1$ for convenience), recover the comparison bits $b_i = \mathbf{1}\{a_i < m_i\}$ simultaneously, and retain the indicated half of each nonsingleton interval. After $\lceil \log_2 n\rceil$ rounds, all coordinates of $\boldsymbol{a}$ are determined. 

It takes $O(\log \log n)$ repetitions to amplify the success probability of each wildcard-search subroutine to $1 - O(1 / \log n)$, which suffices to union bound the success probability of $\lceil\log n\rceil$ rounds.  Hence the total quantum query complexity is $O(\sqrt{k}\log k\log n\log\log n)$ when $\epsilon$ is constant.

Furthermore, we can prove that this query upper bound is near optimal. More specifically, we prove
$$
Q_\epsilon(\operatorname{Fix}^{\text{BPR}}_{n, k}) = \widetilde{\Omega}\left(\sqrt k\right).
$$
We first restrict the hard instance family to $\{\boldsymbol{f}_{(n - 1)\boldsymbol{a}} : \boldsymbol{a}\in [2]^k\}$, and denote the resulting problem by $\widetilde{\operatorname{Fix}}_{n, k}$. Then, we show how an quantum algorithm for $\widetilde{\operatorname{Fix}}_{n,k}$ can be used to solve \(k\)-bit search-with-wildcard problem.

For any binary string $\boldsymbol{a}\in [2]^k$, given the following wildcard oracle access to $\boldsymbol{a}$: for $S\subseteq [k], \boldsymbol{x}\in [2]^{|S|}, z\in \{0, 1\}$
$$
\mathbf{O}_{\boldsymbol{a}}\ket{S, \boldsymbol{x}, z} = \ket{S, \boldsymbol{x}, z\oplus \mathds{I}\{\boldsymbol{a}_S = \boldsymbol{x}\}},
$$
we show how to simulate $\mathbf{O}_{\boldsymbol{f}_{(n - 1)\boldsymbol{a}}}$ using access to $\mathbf{O}_{\boldsymbol{a}}$. Throughout the remainder of this section, whenever $S\subseteq[k]$, we regard the elements of $S$ as listed in increasing order. It suffices to find the first coordinate $i$ where $x_i > 0, a_i = 0$ and the first coordinate $j$ where $x_j < n - 1, a_j = 1$. These coordinates can be found with binary search as in Algorithm \ref{alg:find-first-1} and Algorithm \ref{alg:find-first-0} respectively, and $\mathbf{O}_{\boldsymbol{f}_{(n - 1)\boldsymbol{a}}}$ can be implemented using Algorithm \ref{alg:hard-function}.
\begin{figure}[H]
\noindent
\begin{minipage}[t]{0.48\textwidth}
\vspace{0pt}
\begin{algorithm}[H]
\caption{$\textsc{FindFirstZero}(\boldsymbol{x})$}
\label{alg:find-first-1}
\begin{algorithmic}[1]
    \Require{$\boldsymbol{x}\in [n]^k$};
    \State $S \leftarrow \{i \in [k]:x_i > 0\}$;
    \State $l\leftarrow 1, r\leftarrow |S|, p\leftarrow 0$;
    \While{$l \leq r$}
        \State $m \leftarrow \lfloor(l + r) / 2\rfloor$;
        \State $Q \leftarrow \textsf{first $m$ elements in $S$}$;
        \If{$\mathbf{O}_{\boldsymbol{a}}(Q, \boldsymbol{1}) = 1$}
            \State $p\leftarrow m, l\leftarrow m + 1$;
        \Else
            \State $r\leftarrow m - 1$;
        \EndIf
    \EndWhile
    \If{$p = |S|$}
        \State \textbf{return} $\textsf{None}$;
    \Else
        \State \textbf{return} \textsf{the $p + 1$-th element in $S$};
    \EndIf
\end{algorithmic}
\end{algorithm}
\end{minipage}
\hfill
\begin{minipage}[t]{0.48\textwidth}
\vspace{0pt}
\begin{algorithm}[H]
\caption{$\textsc{FindFirstOne}(\boldsymbol{x})$}
\label{alg:find-first-0}
\begin{algorithmic}[1]
    \Require{$\boldsymbol{x}\in [n]^k$};
    \State $S \leftarrow \{i \in [k]:x_i < n - 1\}$;
    \State $l\leftarrow 1, r\leftarrow |S|, p\leftarrow 0$;
    \While{$l \leq r$}
        \State $m \leftarrow \lfloor(l + r) / 2\rfloor$;
        \State $Q \leftarrow \textsf{first $m$ elements in $S$}$;
        \If{$\mathbf{O}_{\boldsymbol{a}}(Q, \boldsymbol{0}) = 1$}
            \State $p\leftarrow m, l\leftarrow m + 1$;
        \Else
            \State $r\leftarrow m - 1$;
        \EndIf
    \EndWhile
    \If{$p = |S|$}
        \State \textbf{return} $\textsf{None}$;
    \Else
        \State \textbf{return} \textsf{the $p + 1$-th element in $S$};
    \EndIf
\end{algorithmic}
\end{algorithm}
\end{minipage}
\end{figure}

\begin{algorithm}[H]
\caption{$\boldsymbol{f}_{(n - 1)\boldsymbol{a}}(\boldsymbol{x})$}
\label{alg:hard-function}
\begin{algorithmic}[1]
    \Require $\boldsymbol{x}\in [n]^k$;
    \State $i\leftarrow \textsc{FindFirstZero}(\boldsymbol{x}), j\leftarrow \textsc{FindFirstOne}(\boldsymbol{x})$;
    \State $\boldsymbol{y} = \boldsymbol{x}$;
    \If{$i\ne \textsf{None}$}
        \State $y_i = y_i - 1$;
    \EndIf
    \If{$j \ne \textsf{None}$}
        \State $y_j = y_j + 1$;
    \EndIf
    \State \textbf{return} $\boldsymbol{y}$;
\end{algorithmic}
\end{algorithm}

We analyze Algorithm \ref{alg:find-first-1} as an example; the analysis of Algorithm \ref{alg:find-first-0} is analogous. On line 6, $\mathbf{O}_{\boldsymbol{a}}(Q, \boldsymbol{1})$ returns $1$ if the first $m$ entries of $S$ satisfy $a_i = 1$. Algorithm \ref{alg:find-first-1} is a standard binary search procedure whose answer is stored in $p$, hence $p$ stores the length of the longest prefix of $S$ all of whose elements satisfy $a_i = 1$, and $p = 0$ if no such prefix exists. Consequently, the correctness of Algorithm \ref{alg:find-first-1} reduces to the following case analysis:
\begin{enumerate}
    \item If $p = 0$ and $S\neq \varnothing$, then the minimal element $i$ of $S$, which satisfies $x_i > 0$ and $a_i = 0$, is the required coordinate. The procedure returns the smallest element in $S$, as intended.
    \item If $0 < p < |S|$, then the first $p$ elements in $S$ satisfy $x_i > 0$, $a_i = 1$. However, the $p + 1$-th element of $S$ satisfies $x_i > 0$, $a_i = 0$, which is returned as intended.
    \item If $p = |S|$, then $a_i = 1$ for all $i\in S$, equivalently, for every coordinate with $x_i > 0$. Hence, the required coordinate does not exist, and the procedure returns $\textsf{None}$.
\end{enumerate}

Together with the analysis above, it is immediate that Algorithm \ref{alg:hard-function} implements $\boldsymbol{f}_{(n - 1)\boldsymbol{a}}$ with $O(\log k)$ wildcard queries. To implement the simulation coherently, we run the binary search procedure for a fixed $O(\log(k+1))$ number of iterations, using an active flag and dummy queries to pad branches that terminate early, and then reverse the computation to uncompute all workspace registers. The number of wildcard queries remains $O(\log k)$. The unique fixed point of $\boldsymbol f_{(n-1)\boldsymbol a}$ is $(n-1)\boldsymbol a$, from which \(\boldsymbol a\) can be recovered without any additional queries. Consequently, an algorithm that computes $\widetilde{\operatorname{Fix}}_{n, k}$ in $T$ queries would compute $k$-bit $\textsc{WildcardSearch}$ within $O(T\log k)$ queries. Thus,
$$
Q_\epsilon(\textsc{WildcardSearch}_k) \leq O\left(\log k \cdot Q_\epsilon(\widetilde{\operatorname{Fix}}_{n, k})\right).
$$
Since $Q_\epsilon(\textsc{WildcardSearch}_k) \geq \Omega(\sqrt{k})$, it follows that
$$
Q_\epsilon(\operatorname{Fix}^{\text{BPR}}_{n, k}) \geq Q_\epsilon(\widetilde{\operatorname{Fix}}_{n, k}) = \Omega\left(\frac{\sqrt{k}}{\log k}\right).
$$

\section{Alternate Proof of Lower Bound for Binary Search} \label{apdx:binary-search}

We illustrate the tree-filtration adversary method by deriving the standard $\Omega(\log n)$ quantum query lower bound for binary search. The input instances are functions $f_i \colon [n]\rightarrow [2]$, indexed by $i\in [n]$, where
$$
f_i(x) = \begin{cases}
    1, & x \leq i, \\
    0, & x > i.
\end{cases}
$$

The task is to identify the threshold $i$.

It suffices to consider $n$ that is a power of $2$: for general $n$, we may restrict to the first $2^{\lfloor\log_2 n\rfloor}$ instances. Write $n = 2^d$, and represent every $x\in [n]$ by a $d$-bit string $x_0\cdots x_{d - 1}$, with the most significant bit first. Let $\mathcal{T}$ be the complete binary trie of these representations. A node $u$ at depth $t$ is identified with a prefix $u_0\cdots u_{t - 1}$, and its descendant leaves represent exactly the thresholds whose binary expansions begin with this prefix. Define

$$
\hat{D}(u) = \{x\in [n] : x_0\cdots x_{t-1} = u_0\cdots u_{t - 1}\}.
$$

\begin{lemma}
    $\hat{D}$ is a simplified difference criterion for $\mathcal{T}$.
\end{lemma}

\begin{proof}
    Let $L_{u}$ and $R_u$ be, respectively, the smallest and largest integers whose binary expansions extend $u_0\cdots u_{t - 1}$. Thus the instances below $u$ are precisely $f_i$ with $L_u \leq i \leq R_u$. If $x\notin \hat{D}(u)$, then either $x < L_u$, in which case $f_i(x) = 1$ for every such $i$, or $x > R_u$, in which case $f_i(x) = 0$ for every such $i$. Therefore for all instances $f_i, f_j$ below $u$,
    $$
    D(f_i, f_j)\subseteq \hat{D}(u).
    $$

    Moreover extending a binary prefix can only shrink the corresponding set. So $\hat{D}(v)\subseteq \hat{D}(u)$ whenever $v$ is a descendant of $u$.
\end{proof}

Let $\mathcal{A}$ be the top-down random walk on $\mathcal{T}$ that chooses each child with probability $1/2$.

\begin{lemma}
    $\mathcal{A}$ satisfies the $\boldsymbol{a}$-anti-concentration property with $a_t = 1 - \frac{1}{2^{d-t}}$ for $0\leq t < d$.
\end{lemma}

\begin{proof}
    Conditioned on reaching a node at depth $t$, the random walk is uniform over its $2^{d - t}$ descendant leaves, which correspond to distinct thresholds.  Hence, for any fixed descendant leaf, the probability that the random leaf has a different output is
    $$
    1 - 2^{-(d - t)},
    $$
    whenever $t < d$.
\end{proof}

\begin{lemma}
    $\mathcal{A}$ satisfies the $1/2$-geometric-indistinguishability property.
\end{lemma}

\begin{proof}
    Fix depths $0 \leq t < s \leq d$, a node $u$ at depth $t$, and a query coordinate $x\in [n]$. If $x\notin \hat{D}(u)$, then $x\notin \hat{D}(v)$ for every descendant $v$ of $u$. Otherwise, for a random descendant $V$ at depth $s$, the event $x\in \hat{D}(V)$ holds exactly when its next $s - t$ bits agree with those of $x$. Since these bits are independent and uniform, 
    $$
    \Pr_{V\xleftarrow{s}\mathcal{A}}\left[x\in \hat{D}(V) \middle| \mathcal{A}\xrightarrow{t} u\right] = 2^{-(s - t)}.
    $$

    The required upper bound follows with $r = 1/2$.
\end{proof}

Applying Theorem \ref{thm:tree-adversary} with $r = 1/2$ and $a_t = 1 - 2^{-(d-t)}$, we obtain
$$
Q_\epsilon(\mathrm{BinarySearch}_n) \geq \left(1 - 2\sqrt{\epsilon(1 - \epsilon)}\right)\frac{1 - \sqrt{1 / 2}}{1 + \sqrt{1 / 2}}\left({\log_2 n} - 1 + \frac{1}{n}\right) = \Omega(\log n).
$$

Thus the tree-filtration framework recovers the optimal asymptotic lower bound without invoking the spectral analysis of the Hilbert matrix.

\end{document}